\documentclass{IEEEtran}

\usepackage[T1]{fontenc}
\usepackage[margin=1in]{geometry}

\usepackage{amsmath}
\usepackage{amssymb}
\usepackage{amsthm}
\usepackage{amsfonts}
\usepackage{mathtools}
\usepackage[mathscr]{eucal}

\usepackage{adjustbox}
\usepackage{afterpage}
\usepackage{booktabs}
\usepackage{colortbl}
\usepackage{comment}
\usepackage{graphicx}
\usepackage{nicematrix}
\usepackage[subrefformat=parens]{subcaption}
\usepackage{xcolor}

\usepackage{tikz}
\usepackage{tikz-qtree}
\usepackage{tikz-cd}
\usetikzlibrary{arrows}
\usetikzlibrary{arrows.meta}
\usetikzlibrary{backgrounds}
\usetikzlibrary{calc}
\usetikzlibrary{decorations.pathreplacing}
\usetikzlibrary{decorations.pathmorphing}
\usetikzlibrary{intersections}
\usetikzlibrary{fit}
\usetikzlibrary{patterns}
\usetikzlibrary{positioning}
\usetikzlibrary{shapes}
\usetikzlibrary{shadows}
\usetikzlibrary{shapes.geometric}
\usetikzlibrary{tikzmark}

\usepackage[style=ieee,citestyle=ieee-comp,natbib=true,maxcitenames=2,mincitenames=1]{biblatex}
\usepackage{newfloat}
\DeclareFloatingEnvironment[name=Listing]{listing}
\usepackage[inline]{enumitem}
\newlist{inlinelist}{enumerate*}{1}
\setlist[enumerate]{wide=\parindent,labelsep=.5\parindent,leftmargin=*}
\setlist[itemize]{wide=\parindent,labelsep=.5\parindent,leftmargin=*}

\usepackage[hypertexnames=false]{hyperref}
\hypersetup{colorlinks=true,citecolor=blue,linkcolor=black}
\usepackage{doi}

\usepackage{algorithm}
\usepackage[indLines=true,noEnd=true,commentColor=darkgray,endLComment={}]{algpseudocodex}
\algrenewcommand\algorithmicrequire{\textbf{Input:}}
\algrenewcommand\algorithmicensure{\textbf{Output:}}

\makeatletter
\AddToHook{env/algorithmic/begin}{\def\@currentcounter{ALG@line}}
\makeatother

\usepackage[nameinlink,capitalize]{cleveref}
\newcommand{\crefnameof}[1]{\csname cref@#1@name\endcsname}

\makeatletter
  \DeclareRobustCommand{\labelcrefrange}[2]{\@crefrangenostar{labelcref}{#1}{#2}}
\makeatother
\crefname{equation}{}{}
\Crefname{equation}{Eq.}{Eqs.}
\crefname{ALG@line}{line}{lines}
\Crefname{ALG@line}{Line}{Lines}
\crefname{inlinelisti}{statement}{statements}
\Crefname{inlinelisti}{Statement}{Statements}

\usepackage{thmtools}
\newcommand{\customqed}{\textup{$\lozenge$}}
\declaretheorem[style=plain,qed=\customqed]{theorem}
\declaretheorem[style=plain,qed=\customqed]{proposition}
\declaretheorem[style=plain,qed=\customqed]{lemma}
\declaretheorem[style=plain,qed=\customqed]{corollary}
\declaretheorem[style=definition,qed=\customqed]{definition}
\declaretheorem[style=definition,qed=\customqed]{example}
\declaretheorem[style=definition,qed=\customqed]{remark}
\declaretheorem[style=definition,qed=\customqed]{problem}

\newcommand{\defeq}{\coloneqq}
\newcommand{\eqdef}{\eqqcolon}

\newcommand{\diff}{\mathrm{d}}
\newcommand{\real}{\mathbb{R}}

\newcommand{\nat}{\mathbb{N}}
\newcommand{\integers}{\mathbb{Z}}
\newcommand{\expect}[1]{\mathbb{E}\left[{#1}\right]}
\newcommand{\cost}[1]{\mathscr{C}\!\left({#1}\right)}
\newcommand{\bp}{P}
\newcommand{\bq}{{Q}}

\newcommand{\bb}{{\mathbf{b}}}
\newcommand{\bc}{{\mathbf{c}}}
\newcommand{\dom}{\mathrm{dom}}
\NewDocumentCommand{\leaves}{smmm}{\ell_{#2}(#3,#4)}
\NewDocumentCommand{\toll}{sm}{\tau\IfBooleanTF{#1}{(#2)}{\!\left(#2\right)}}
\NewDocumentCommand{\trel}{sm}{\tau_{\rm r}\IfBooleanTF{#1}{(#2)}{\!\left(#2\right)}}
\NewDocumentCommand{\trelfldr}{sm}{\tau_{\rm r, \fldr}\IfBooleanTF{#1}{(#2)}{\!\left(#2\right)}}
\NewDocumentCommand{\Hb}{sm}{H_{\rm b}\IfBooleanTF{#1}{(#2)}{\!\left(#2\right)}}
\NewDocumentCommand{\Ho}{sm}{H_{1}\IfBooleanTF{#1}{(#2)}{\!\left(#2\right)}}
\NewDocumentCommand{\nuu}{sm}{\nu\IfBooleanTF{#1}{(#2)}{\!\left(#2\right)}}
\NewDocumentCommand{\epsd}{sO{d}m}{\epsilon_{#2}\IfBooleanTF{#1}{(#3)}{\!\left(#3\right)}}
\newcommand{\flip}{\mathsf{flip}}

\newcommand{\TextOp}[1]{\mathop{\text{\normalfont #1}}\nolimits}
\DeclareMathOperator{\KY}{\TextOp{KY}}
\DeclareMathOperator{\FLDR}{\TextOp{FLDR}}
\DeclareMathOperator{\ALDR}{\TextOp{ALDR}}
\NewDocumentCommand{\ky}{so}{\KY\IfValueT{#2}{\IfBooleanF{#1}{\!}\left(#2\right)}}
\NewDocumentCommand{\fldr}{so}{\FLDR\IfValueT{#2}{\IfBooleanF{#1}{\!}\left(#2\right)}}
\NewDocumentCommand{\aldr}{so}{\ALDR\IfValueT{#2}{\IfBooleanF{#1}{\!}\left(#2\right)}}

\DeclareMathOperator*{\argmin}{arg\,min}
\allowdisplaybreaks

\tikzset{every tree node/.style={anchor=north}}
\tikzset{every leaf node/.style={inner sep=1pt, minimum height=10pt,draw=none}}
\tikzset{branch/.style={shape=coordinate}}
\tikzset{leaf/.style={inner sep=1pt,draw, minimum height=10pt}}

\DeclarePairedDelimiter{\set}{\lbrace}{\rbrace}
\DeclarePairedDelimiter{\ceil}{\lceil}{\rceil}
\DeclarePairedDelimiter{\floor}{\lfloor}{\rfloor}
\DeclarePairedDelimiter{\abs}{\lvert}{\rvert}

\title{Efficient Rejection Sampling in the Entropy-Optimal Range}

\author{Thomas L.~Draper and Feras A.~Saad
  \thanks{The authors are with the Computer Science Department at Carnegie
    Mellon University, 5000 Forbes Avenue, Pittsburgh, PA 15213 USA
    (e-mail:
      \href{mailto:tdraper@cmu.edu}{tdraper@cmu.edu};
      \href{mailto:fsaad@cmu.edu}{fsaad@cmu.edu})}}

\newcommand{\footer}{\copyright 2026 IEEE.  Personal use of this material is permitted.
Permission from IEEE must be obtained for all other uses, in any current or
future media, including reprinting/republishing this material for
advertising or promotional purposes, creating new collective works, for
resale or redistribution to servers or lists, or reuse of any copyrighted
component of this work in other works.
\begin{refsection}
Appears as: \fullcite{draper2026b}
\end{refsection}
}

\usepackage{eso-pic}
\AddToShipoutPictureFG*{%
  \ifnum\value{page}=1\relax
    \AtTextLowerLeft{%
      \raisebox{-12.5mm}[0pt][0pt]{%
        \parbox{\textwidth}{%
          \scriptsize
          \emergencystretch=1em 
          \color{blue}\footer
        }%
      }%
    }%
  \fi
}

\begin{document}

\maketitle

\begin{abstract}
We study the problem of generating a random variate $X$ from a finite discrete
probability distribution $\bp$ using an entropy source of independent
fair coin flips.
A classic result from \citeauthor{knuth1976} shows that the optimal
expected number of input coin flips per output sample lies between $H(\bp)$
and $H(\bp)\,{+}\,2$, where $H$ is the Shannon entropy function.
However, implementing the \citeauthor{knuth1976} ``entropy-optimal''
sampler entails a tradeoff between using either exponential space with low
runtime per sample, or linear space with high runtime
per sample.
We introduce a new sampling algorithm that avoids this
tradeoff: it requires linearithmic space, incurs negligible runtime
overhead per sample, and uses an expected number of coin flips
that lies in
the entropy-optimal range $[H(\bp), H(\bp)\,{+}\,2)$.
No previous sampler for discrete distributions simultaneously
achieves these space, time, and entropy characteristics.
Numerical experiments demonstrate improvements in runtime and
entropy of the proposed method compared to the celebrated alias method.
\end{abstract}


\section{Introduction}
\label{sec:introduction}

We are concerned with algorithms that use independent flips of a fair coin
to generate a random variate $X$ from a rational discrete probability distribution
$\bp \defeq (p_1, \dots, p_n)$ over $n$ outcomes.
The ``efficiency'' of a sampling algorithm is measured in terms of its
space requirements, running time, and expected number of coin flips drawn from the
entropy source per output (henceforth, ``entropy cost'').
Entropy sources that provide coin flips stem from natural phenomena~\citep{faran1968}---such as
electrical, atmospheric, or thermal noise--or from
software-based mechanisms that collect system-level entropy to seed
a uniform pseudorandom number generator~\citep{Kneusel2018}.

\Citet{knuth1976} settle the problem of developing an ``entropy-optimal''
sampling algorithm for a discrete distribution $\bp$ that consumes the least
possible number of coin flips per output on average.
Its entropy cost is between $H(\bp)$ and $H(\bp)+2$
coin flips per output, where $H(\bp) \defeq \sum_{i=1}^{n}p_i \log(1/p_i)$ is
the Shannon entropy.
The \citeauthor{knuth1976} method is well known in the information theory
literature~\citep[\S5.11]{cover2012}.
On the practical side, it has been used
to develop high-performance and cryptographically secure hardware
devices for nonuniform random number
generation~\cite{roy2013,dwarakanath2014,follath2014,baidya2024}, where
entropy is a premium resource and expensive numerical operations cannot be
performed during sampling time.

Despite these advantages, broader adoption of the \citeauthor{knuth1976}
method has been limited by practical implementation difficulties.
With the notable exception of \citet[Chapter~15]{devroye1986},
entropy-optimal sampling is largely absent from standard references on
random variate generation (e.g.,
\cites[\S3.3]{ripley1987}[Chapter~3]{fishman1996}[Chapters 3 and 10]{hormann2004}{schwarz2011}),
and it is not available in prominent numerical
software libraries \citep{leydold2009,galassi2009,lea1992}.
This gap reflects a lack of implementations that achieve the entropy
guarantee while also satisfying the space and time constraints demanded by
modern applications.

The standard implementation strategy is to first construct, in a preprocessing
phase, a data structure called a discrete distribution generating (DDG)
tree, and then to generate samples by traversing this tree.
While this strategy enables fast sampling once the tree is constructed,
\citet[Theorem~3.5]{saad2020popl} show that the space complexity of
entropy-optimal DDG trees scales exponentially with the number of bits
used to encode $\bp$.
Explicit DDG-tree construction is thus inefficient for target
distributions specified at moderate-to-high precision.

An alternative implementation strategy~\cites[Alg.~1]{roy2013}[Alg.~A1]{draper2026}
is to avoid explicitly constructing a DDG tree during preprocessing
by instead using lazy evaluation during sampling.
In particular, the exponential space cost is circumvented by
constructing only those parts of the DDG tree that are demanded by the coin
flips revealed during sampling.
However, the resulting on-the-fly construction incurs a high runtime cost
per sample, diminishing the benefits of entropy-optimal
sampling.

These considerations underscore a fundamental tension among three
resources---space, time, and entropy---and motivate the central question of
the present work: \textit{is it possible to design a space- and
time-efficient sampler whose entropy cost lies in the optimal range of at
most $H(P)+2$ coin flips per output sample?}

\subsection{Main Results}

\begin{table*}[!t]
\centering
\caption{
  Comparison of exact sampling methods for discrete probability
  distributions $\bp \defeq (a_1/m, \dots, a_n/m)$.
  As the input $\bp$ is represented using $n \log(m)$ bits, an expression
  with space complexity $O(m)$ scales exponentially in the input size
  (also called pseudo-polynomial).
  The entropy cost of any method is at least
  $H(\bp) \leq \log(n)$.
  For algorithms that explicitly traverse a discrete distribution
  generating (DDG) tree, the expected sampling runtime is equal to
  the entropy cost.
  For algorithms that implicitly traverse an underlying DDG tree, the
  expected sample time may be higher than the entropy cost,
  because operations such as binary search are performed during sampling.
  All bounds in this table are tight; in particular, the Amplified
  Loaded Dice Roller (\cref{alg:aldr}) is the only method that simultaneously
  achieves entropy cost $<H(\bp) + 2$ and linearithmic space with no runtime
  overhead, for every distribution $\bp$.
  The space complexity has an asymptotic upper bound, the entropy
  cost has a strict upper bound, and the expected sampling time has an
  asymptotic upper and lower bound.
}
\label{table:complexity}
\begin{adjustbox}{max width=\linewidth}
\newcommand{\specialcell}[2][c]{\begin{tabular}[#1]{@{}c@{}}#2\end{tabular}}
\begin{tabular}{|llllll@{}|}
\hline
\textbf{Method}
& \textbf{DDG Tree}
& \textbf{Space}
& \textbf{Entropy Cost} $C$
& \textbf{Runtime}
& \textbf{Reference} \rule{0pt}{2ex}
\\ \hline \hline
Entropy-Optimal
& Explicit
& $nm\log(n)$
& $< H(\bp) + 2$
& $C$
& \Citet[\S2]{knuth1976} \rule{0pt}{2.5ex}
\\ 
\bfseries \begin{tabular}[t]{@{}l}Rejection Sampling\\\quad (Amplified Loaded Dice Roller)\end{tabular}
& Explicit
& $n\log(m)\log(n)$
& $< H(\bp) + 2$
& $C$
& \bfseries {\cref{alg:aldr}}
\\ 
\begin{tabular}[t]{@{}l}Interval Algorithm\end{tabular}
& Explicit
& $nm\log(n)$
& $< H(\bp) + 3$
& $C$
& \begin{tabular}[t]{@{}l}\Citet{han1997}\\\Citet{gill1972}\end{tabular}
\\ 
\begin{tabular}[t]{@{}l}Rejection Sampling\\\quad (Fast Loaded Dice Roller)\end{tabular}
& Explicit
& $n\log(m)\log(n)$
& $< H(\bp) + 6$
& $C$
& \Citet[Alg.~4]{saad2020fldr}
\\ \hline
Entropy-Optimal
& Implicit
& $n \log(m)$
& $< H(\bp) + 2$
& $C n \log(m)$
& \begin{tabular}[t]{@{}l}\Citet[Alg.~1]{roy2013}\\\Citet[Alg.~A1]{draper2026}\end{tabular}
\rule{0pt}{2.5ex}
\\
Alias Method
& Implicit
& $n\log(m)$
& $< \ceil{\log(n)}+3$
& $C$
& \Citet{walker1977}; \Citet{vose1991}
\\ 
\begin{tabular}[t]{@{}l}Interval Algorithm\end{tabular}
& Implicit
& $n\log(m)$
& $< H(\bp) + 3$
& $C\log(n)$
& \begin{tabular}[t]{@{}l}\Citet{han1997}\\\Citet{uyematsu2003}\end{tabular}
\\ 
\begin{tabular}[t]{@{}l}Rejection Sampling\\\quad (Uniform Proposal)\end{tabular}
& Implicit
& $\log(m)$
& $< n (\ceil{\log(n)} + 3)$
& $C$
& \begin{tabular}[t]{@{}l}\Citet[Alg.~1]{saad2020fldr}\\\Citet{lumbroso2013}\end{tabular}
\\ 
\begin{tabular}[t]{@{}l}Rejection Sampling\\\quad(Dyadic Proposal; Lookup Table)\end{tabular}
& Implicit
& $m\log(n)$
& $< 2\ceil{\log(m)}$
& $C$
& \begin{tabular}[t]{@{}l}\Citet[Alg.~2]{saad2020fldr}\\\Citet[p.~770]{devroye1986}\end{tabular}
\\ 
\begin{tabular}[t]{@{}l}Rejection Sampling\\\quad (Dyadic Proposal; Binary Search)\end{tabular}
& Implicit
& $n\log(m)$
& $< 2\ceil{\log(m)}$
& $\log(n) + C$
& \begin{tabular}[t]{@{}l}\Citet[Alg.~3]{saad2020fldr}\\\Citet[p.~770]{devroye1986}\end{tabular}
\\ \hline
\end{tabular}
\end{adjustbox}
\end{table*}

\begin{definition}[{\citep[Definition~4]{han1993}}]
\label{def:m-type}
For any integer $M \ge 1$, a discrete probability distribution $\bp \defeq
(p_1,\dots,p_n)$ over $n \geq 1$ outcomes is said to be ${M}$-\textit{type}
if each $p_i = {A}_i/{M}$
for some integer ${A}_i \in \set{0,\dots,{M}}$ ($1 \le i \le n$).%
\end{definition}

Any finite discrete probability distribution with rational
probabilities---say $p_i = d_i/b_i$ in lowest terms $(1 \le i \le n)$---%
is $cm$-type for each $c\ge 1$, where $m$ is the least common multiple of $b_1, \dots, b_n$
and the numerators are $A_i = (cm/b_i)d_i$.
If $c=1$ then all the numerators $A_i$ are coprime,
yielding the $m$-type representation
of $\bp$ with the smallest possible denominator.

\begin{problem}
\label{problem:sampling}
Given an $m$-type
probability distribution $\bp$
encoded as a list $(a_1, \dots, a_n)$ of $n$ coprime
positive integers
and access to a stream of independent fair coin flips, generate an
integer $i$ with probability $p_i = a_i/m$.
\end{problem}

This article introduces the \textit{Amplified Loaded Dice Roller} ($\aldr$),
a family of samplers for $m$-type distributions
that combine the entropy-optimal
method of \citet{knuth1976} with the rejection sampling method of
\citet{neumann1951}.
Each $\aldr$ sampler for $\bp$ is parameterized by an integer ``amplification'' parameter $K \geq k
\defeq \ceil{\log(m)}$ that governs its space and entropy cost, written $\aldr[\bp,K]$.
The results are as follows.

\begin{itemize}[leftmargin=*,itemsep=5pt]
\item \Cref{theorem:aldr-generic-bound}
shows that $\aldr[\bp,K]$ has an entropy cost bounded by $H(\bp)+2+O((K-k)/2^{K-k})$.

\item \Cref{theorem:aldr-2k-toll-two}
shows that $\aldr[\bp,K]$ has an entropy cost less than $H(\bp)+2$ for every $K \geq 2k$.

\item \Cref{theorem:aldr-doubling-minimal} shows that the smallest
parameter that ensures the $H(\bp)+2$ bound is $K=2k$,
achieving linearithmic space complexity of $n \log(m)\log(n)$
with respect the $n\log(m)$-sized input.

\item \Cref{theorem:aldr-ky-match-bounded}
describes the distributions $\bp$ for which $\aldr[\bp,K]$
is entropy optimal for some $K$.

\item \Cref{theorem:aldr-leq-fldr}
shows that the entropy cost of each $\aldr[\bp,K]$ sampler
($K \geq k$) is upper bounded by that of the first member ($K = k$),
giving the exact conditions for strict inequality.

\end{itemize}


\begin{table*}[t]
\centering
\normalsize
\caption{Variations of the random number generation problem investigated in
the literature, under various assumptions on the input source and generated
output variates.
The assumptions in this article are underlined (the proposed
method is readily generalizable to an input source giving
i.i.d.~rolls of a fair $k$-sided dice).
}
\label{table:assumptions}
\begin{adjustbox}{max width=\linewidth}
\begin{tabular}{|c|c|c||c|c|c|c|}
\hline
\multicolumn{3}{|c||}{\textbf{Input Source}} & \multicolumn{4}{c|}{\textbf{Output Variates}} \\ \hline\hline
\textbf{Symbols}                             & \textbf{Distribution}                                         & \textbf{Sequence}  & \textbf{Symbols}                     & \textbf{Distribution} & \textbf{Error}    & \textbf{Length} \\ \hline
\underline{Coins} $(\set{0,1})$              & \underline{Uniform}                                           & \underline{i.i.d.} & Coins $(\set{0,1})$                  & Uniform               & \underline{Exact} & \underline{Fixed-length} \\
Dice $(\set{1,\dots,k})$                     & Arbitrary (Known)                                             & Markov             & \underline{Dice} $(\set{1,\dots,n})$ & \underline{Arbitrary} & Approximate       & Variable-length \\
~                                            & Arbitrary (Unknown)                                           & Nonstationary      & ~                                    & ~                     & ~                 & ~  \\ \hline
\end{tabular}
\end{adjustbox}
\end{table*}

In particular, when $K=2k$, the entropy cost lies in the optimal
range of $[H(\bp), H(\bp)+2)$ coin flips per sample, using space that scales
linearithmically with the input size of $n\log(m)$ bits used to specify
the $m$-type distribution $\bp$, and runtime that matches the entropy cost
with no additional overhead.
No existing sampler
simultaneously achieves these space, runtime, and entropy characteristics
for \textit{every} distribution $\bp$.
The preprocessing and generation phases of $\aldr$ are
readily implementable using fast integer
arithmetic and a simple data structure.
The method is well suited for any situation that requires exact samples.
It is also suitable for
sampling on a constrained hardware device, where floating-point units
are unavailable or introduce unacceptable errors, and where the overhead of
arbitrary-precision (bignum) arithmetic is prohibitively high.
\Cref{table:complexity} compares the computational complexity
of $\aldr$ (at $K=2k$) with various existing sampling algorithms in the
literature, which are described next.

\subsection{Related Sampling Techniques}

\paragraph*{Exact Samplers}

Several works have developed concrete algorithms for entropy-optimal
sampling~\citep{knuth1976}, under various assumptions.
\Citet{lumbroso2013} describes an efficient, linear space implementation
of the entropy-optimal \citeauthor{knuth1976} method when $\bp$ is the
uniform or Bernoulli distribution.
\Citet{huber2024} analyze this optimal uniform sampler as a
``randomness recycler'' protocol and generalize it to arbitrary discrete
distributions, assuming access to the binary expansions of the target
probabilities.
\Textcites[Alg.~1]{roy2013}[Alg.~A1]{draper2026} give a lazy, space-efficient
implementation of entropy-optimal sampling that constructs the DDG tree on
the fly, but incurs a high runtime overhead from computing the binary
expansions of probabilities during sampling.
\Citet{knuth1976,saad2025} show how to implement entropy-optimal sampling
given access to a procedure that computes the cumulative distribution
function of $\bp$, which can handle a much larger number of outcomes
as compared to being given access to an array of weights.
The table method of \citet{marsaglia1963} is entropy optimal,
but applies only when each probability in $\bp$ is dyadic
(i.e., the denominators of the $p_i$ are powers of $2$).
\Citet[\S2.1.1]{devroye2020} discuss further practical considerations for
implementing entropy-optimal generators in software.

\Citet{gill1972} describes an exact sampling algorithm for discrete
distributions, based on a lazy implementation of the inverse transform
method, and proves that it consumes at most $H(\bp)+4$ flips per output on
average.
This algorithm is a special case of the more general interval algorithm of
\citet{han1997}, who establish a tighter bound of $H(\bp)+3$.
Typical implementations of the interval algorithm
use linear space but perform an expensive $O(\log n)$ binary search
after each coin flip is obtained from the source~\citep{uyematsu2003,devroye2020}.
Eliminating the binary search at sampling time is possible by
constructing an exponentially sized tree data structure during preprocessing,
or by navigating other space--time trade-offs.
\Citet{saad2020fldr} present a linearithmic space sampler, based on
combining entropy-optimal sampling and rejection sampling, whose
entropy cost is less than $H(\bp)+6$.
The family of samplers introduced in this article generalize and
improve upon this prior work.


\begin{figure*}[!t]
\tikzset{level distance=20pt, sibling distance=10pt}
\tikzstyle{branchE}=[font=\footnotesize,fill=white,inner sep=0.3pt]
\captionsetup[subfigure]{aboveskip=0pt,belowskip=5pt}
\centering
\begin{subfigure}[t]{.2\linewidth}
\caption{$\bp = (1/4, 3/4)$}
\centering
\begin{tikzpicture}
\Tree[
  \edge[-latex] node[branchE,pos=.5]{0}; \node[leaf]{2};
  \edge[-latex] node[branchE,pos=.5]{1} ;
    [
    \edge[-latex] node[branchE,pos=.5]{0}; \node[leaf]{1};
    \edge[-latex] node[branchE,pos=.5]{1}; \node[leaf]{2};
    ] ]
\end{tikzpicture}
\medskip

$\left\lbrace
\begin{aligned}
T(0)&=2\\
T(10)&=1\\
T(11)&=2
\end{aligned}
\right\rbrace
$
\end{subfigure}
\begin{subfigure}[t]{.2\linewidth}
\caption{$\bp = (1/3, 2/3)$}
\centering

\begin{tikzpicture}
\Tree[.\node[branch,name=root]{};
  \edge[-latex] node[pos=.5,branchE]{0};
    [
     \edge[-latex] node[pos=.5,branchE]{0}; \node[leaf,draw,name=back]{\phantom{0}};
     \edge[-latex] node[pos=.5,branchE]{1}; \node[leaf,draw]{1};
     ]
  \edge[-latex] node[pos=.5,branchE]{1}; \node[leaf,draw]{2};
  ]
\draw[red,-latex,out=110,in=180] (back.center) to (root.south);
\end{tikzpicture}
\medskip

$\left\lbrace
\begin{aligned}
T(1)&=2\\
T(01)&=1\\
T(00\bb)&=T(\bb)
\end{aligned}
\right\rbrace
$
\end{subfigure}
\begin{subfigure}[t]{.25\linewidth}
\caption{$\bp = (1/4, 1/4, 1/4, 1/4)$}
\centering
\begin{tikzpicture}
\Tree[
  \edge[-latex] node[pos=.5,branchE]{0};
    [
      \edge[-latex] node[pos=.5,branchE]{0}; \node[leaf,draw]{1};
      \edge[-latex] node[pos=.5,branchE]{1}; \node[leaf,draw]{2};
    ]
  \edge[-latex] node[pos=.5,branchE]{1};
    [
    \edge[-latex] node[pos=.5,branchE]{0}; \node[leaf,draw]{3};
    \edge[-latex] node[pos=.5,branchE]{1}; \node[leaf,draw]{4};
    ] ]
\end{tikzpicture}
\medskip

$\left\lbrace
\begin{aligned}
T(00)&=1\\
T(01)&=2\\
T(10)&=3\\
T(11)&=4
\end{aligned}
\right\rbrace$
\end{subfigure}
\begin{subfigure}[t]{.325\linewidth}
\caption{$\bp = (1/6, 2/6, 2/6, 1/6)$}
\centering
\begin{tikzpicture}
\Tree[
  \edge[-latex] node[pos=.5,branchE]{0};
    [.\node[branch,name=root1]{};
      \edge[-latex] node[pos=.5,branchE]{0};
        [
          \edge[-latex] node[pos=.5,branchE]{0}; \node[leaf,draw,name=back1]{\phantom{1}};
          \edge[-latex] node[pos=.5,branchE]{1}; \node[leaf,draw,name=back2]{\phantom{4}};
        ]
      \edge[-latex] node[pos=.5,branchE]{1};
        [
          \edge[-latex] node[pos=.5,branchE]{0}; \node[leaf,draw]{1};
          \edge[-latex] node[pos=.5,branchE]{1}; \node[leaf,draw]{4};
        ]
    ]
  \edge[-latex] node[pos=.5,branchE]{1};
    [.\node[branch,name=root2]{};
    \edge[-latex] node[pos=.5,branchE]{0}; \node[leaf,draw]{2};
    \edge[-latex] node[pos=.5,branchE]{1}; \node[leaf,draw]{3};
    ]
  ]
\draw[red,-latex,out=110,in=180] (back1.center) to (root1.south);
\draw[red,-latex,out=70,in=180] (back2.center) to (root2.south);
\end{tikzpicture}
\medskip

$\left\lbrace
\begin{aligned}
\begin{array}[t]{r@{\,}lr@{\,}l}
T(10)&=2 & T(11)&=3 \\
T(010)&=1 & T(011)&=4 \\
\multicolumn{4}{c}{T(000\bb)=T(0\bb)}\\
\multicolumn{4}{c}{T(001\bb)=T(1\bb)}
\end{array}
\end{aligned}
\right\rbrace$
\end{subfigure}

\captionsetup{skip=5pt}
\caption[DDG representations]{DDG tree representations of four random sampling algorithms
$T:\set{0,1}^* \rightharpoonup \mathbb{N}$ with output distributions $\bp$.
These trees are constructed using the entropy-optimal \citeauthor{knuth1976}
method from \cref{theorem:knuth-yao}.
The string $\bb \in \set{0,1}^*$ ranges over all finite-length bit string
continuations.
Any string $\bb$ that does not index a path to a leaf node is
not in the domain of $T$.
In future figures of DDG trees, \begin{inlinelist}[label=(\roman*)]
  \item arrows and labels along edges are omitted; and
  \item leaves are labeled as $a_i$ instead of $i$
    for clarity when there are no duplicate weights.
  \end{inlinelist}
}
\label{fig:ddg-tree-examples}
\end{figure*}

\paragraph*{Approximate Samplers}

Most sampling algorithms for discrete distributions used in
practice are based on the so-called ``real RAM'' model of
computation~\citep{shamos1978,blum2004,feferman2013}.
A survey of these
techniques is given in \citet{schwarz2011}.
In the real RAM model, a sampling algorithm is assumed to be able to perform
the following operations in constant time~\citep[Assumptions 1--3]{devroye1986}:
\begin{inlinelist}[label=(\roman*)]
\item obtain i.i.d.\ draws of continuous uniform random variables in
  $[0,1]$;
\item store and look up infinitely precise real numbers;
\item evaluate fundamental real functions with infinite accuracy.
\end{inlinelist}
A sampler is then understood as a map from one or more
uniforms $U_1, U_2, \ldots, U_k$ to an outcome in $\nat$.
For example, a sample from $\bp$ can be
generated using the inverse transform method:
generate $U \sim \mathrm{Uniform}(0,1)$
and then select the integer $i$ that satisfies
$p_1 + \dots + p_{i-1} < U \leq p_1 + \dots + p_i$, using a linear or
binary search through the array of cumulative probabilities.
As this procedure may be too slow for large $n$,
specialized data structures can be constructed
during preprocessing to speed up generation,
such as
the \citeauthor{marsaglia1963} table method~\citep{marsaglia1963,norman1972},
the \citeauthor{chen1974} ``index table'' method~\citep{chen1974},
and the \citeauthor{walker1977} alias method~\citep{walker1977}.
Typical implementations of these samplers in numerical
libraries~\citep{galassi2009,leydold2009}
suffer from many sources of approximation errors~\citep{Monahan1985,Mironov2012,goualard2020,goualard2022,saad2025},
such as using a floating-point uniform $\hat{U} = W \div d$
to approximate the real uniform $U$, where $W$ is a random integer
comprised of 32, 53, or 64 bits and $d$ is a fixed denominator
(e.g., a power of two or Mersenne number).
Replacing floating-point arithmetic with arbitrary-precision arithmetic
could address these errors in principle~\citep{devroye2020},
but imposes large computational overhead in practice.
These implementations also waste entropy, because the
number of coin flips used to generate $W$ may greatly exceed $H(\bp)$
(e.g., if $\bp = \mathrm{Bernoulli}(10^{-3})$,
generating $W$ requires over $1000 \times H(\bp)$ flips).
As compared to approximate samplers that use floating-point arithmetic,
the method in this article is exact and uses fast integer arithmetic.

\paragraph*{Variations}
The random number generation problem has been widely investigated in the
literature under a variety of assumptions on the input source and output
distribution (\cref{table:assumptions}).
These variants are less common in practical applications
than the problem of converting fair coin flips to arbitrary dice rolls.
For example, several authors have studied the problem of extracting fair coin flips
from an i.i.d.\ source with an arbitrary but known
distribution~\citep{elias1972,abrahams1996,roche1991,pae2005a,kozen2014,pae2015,pae2020}.
\Citet{han1997} and \citet{kozen2018} explore more general reductions,
where the former gives an elegant method for converting a sequence of rolls
of an arbitrary $k$-sided dice into rolls of an arbitrary $n$-sided
dice, where the input and output sequences may be i.i.d., Markov, or
arbitrary stochastic processes.
Another variant is extracting fair coin flips from a coin or dice whose
distribution is \textit{unknown}, when the source is
i.i.d.~\citep{neumann1951,hoeffding1970,stout1984,cohen1985,peres1992,pae2006} or a
stationary Markov chain~\citep{elias1972,blum1986}.
Some samplers~\citep{elias1972,peres1992,cicalese2006}
produce a variable-length output instead of a single
fixed-length output, where the number of outputs is random and is
determined by the realized values of the coin flips produced by the entropy source.

Another generalization allows the sampler to produce approximate
samples from $\bp$ up to a given statistical error tolerance, which is
investigated by~\citet{vembu1995,han1993} in the asymptotic regime.
In the non-asymptotic setting, \citet{saad2020popl} show how to find an
$m$-type approximation $\hat{\bp}$ to a given distribution $\bp$ with
$n$ outcomes that achieves minimal
approximation error in terms of any $f$-divergence
using $O(n)$ space and $O(n\log{n})$ time,
generalizing the
results of \citet{Bocherer2016} who considered the total variation and
Kullback--Leibler divergence.
For an irrational distribution $\bp$, our method could be used in
conjunction with the method in \citet{saad2020popl} by first
finding an optimal $m$-type approximation $\hat{\bp}$.
However, in this setting, the denominator $m$ may as well be chosen to be a power of two,
where the usual entropy-optimal \citeauthor{knuth1976} method
on $\hat{\bp}$ achieves linearithmic space.
Therefore, our algorithm is most useful for exactly and efficiently
sampling arbitrary $m$-type distributions.

Finally, several authors develop techniques to reduce the
\textit{amortized} entropy cost of random sampling across a large
number of outputs, either by generating outputs in batches or
by recycling a fraction of the entropy used to generate each output.
Examples include
the distribution generating (DG) tree method in \citet[\S4]{knuth1976},
the iterative interval method in \citet[\S{V}]{han1997},
the randomness extraction method in \citet[\S{2.3}]{devroye2020},
the restart protocols of \citet{kozen2018},
the dice-rolling algorithm in \citet{shao2025},
and the randomness recycling method of \citet{draper2026}.
While we focus on the non-amortized (single-sample) entropy cost, our
method is compatible with general randomness recycling algorithms for DDG
trees that reduce the amortized (multi-sample) entropy cost of random
sampling~\citep[\S5.3]{draper2026}.

\Cref{sec:preliminaries} outlines mathematical preliminaries.
\Cref{sec:fldr} reviews the \textit{Fast Loaded Dice Roller} ($\fldr$) sampler and
provides new results on the tightness of its entropy bound.
\Cref{sec:aldr,sec:properties} introduce the ALDR family of samplers and
characterize its space, time, and entropy complexity.
\Cref{sec:implementation} gives an implementation of $\aldr$ using fast
integer arithmetic and empirically evaluates the algorithm, showing
performance improvements over the widely used \citeauthor{walker1977} alias
method~\citep{walker1977}.
\Cref{sec:remarks} concludes with closing remarks.

\section{Preliminaries}
\label{sec:preliminaries}

To speak precisely about the time, space,
and entropy cost of random sampling algorithms in a realistic model of
computation, we use the ``random bit model'' of random variate
generation~\parencites[Chapter 15]{devroye1986}{knuth1976}.
In contrast to the real RAM model,
which assumes the entropy source provides i.i.d.~uniforms over the real interval $[0,1]$,
the random bit model assumes that the entropy source
provides i.i.d.~fair coin flips,
which implies that a random sampling algorithm can be formally
understood as a binary decision tree.

\subsection{Discrete Distribution Generating Trees}
\label{sec:preliminaries-ddg}

\Citet{knuth1976} introduce \textit{discrete distribution generating} (DDG) trees,
a universal representation of any random sampling algorithm
that maps a sequence of coin flips from an entropy source to an
outcome in $\nat$.
Abstractly, a DDG tree is any partial function
$T: \set{0,1}^* \rightharpoonup \mathbb{N}$ whose domain is a prefix-free set
of finite-length bit strings.
The concrete execution semantics of $T$ can be understood in terms of its
representation as a full binary tree, where each binary string $\bb \in \mathrm{dom}(T)$
in the domain of $T$ indexes a path from the root node to a leaf node labeled $T(\bb)$.
Using this representation, a random number is generated from $T$ as
follows: starting from the root, a coin flip $b \sim \mathrm{Bernoulli}(1/2)$
is drawn from the entropy source.
If $b=0$ (resp.~$b=1$), then the left (resp.\ right) child is visited.
This process repeats until reaching a leaf node, whose label is returned
as the generated sample.
The tree $T$ defines a probability distribution if and only if it is
exhaustive, i.e., $\sum_{b \in \dom(T)}2^{-\abs{b}} = 1$, where
$\abs{b}$ denotes the length of the binary string $b$.
\Cref{fig:ddg-tree-examples} shows four examples of DDG trees.

For any DDG tree $T$, the number of leaves
with label $i \geq 0$ at depth $d \geq 0$ is denoted
$\leaves{T}{d}{i}$.
The partial function $T:\set{0,1}^* \rightharpoonup \mathbb{N}$
can be lifted to a function defined almost-everywhere on $\set{0,1}^\nat$
(i.e., the set of all infinite length binary sequences from the entropy source)
by letting $T(\bb \bc) \defeq T(\bb)$ for all $\bb \in \dom(T)$ and
$\bc \in \set{0,1}^\nat$.
In this way, $T$ is formally understood as an $\nat$-valued discrete random
variable on the standard probability space $([0,1], \mathcal{B}_{[0,1]}, \Pr)$
with \textit{output distribution} $\bp_T \defeq (p_{T,1}, p_{T,2}, \dots)$,
whose probabilities are
\begin{align}
p_{T,i} \defeq \Pr(T = i)
= \sum_{d=0}^{\infty} \leaves{T}{d}{i} 2^{-d}
&& (i \in \nat).
\label{eq:ddg-tree-probs}
\end{align}
The \textit{entropy consumption} of $T$, denoted $\cost{T}$, is a discrete random
variable over $\nat$ that counts the random number of coin flips used by
$T$ to generate an output in a given simulation, whose distribution is
\begin{align}
\Pr(\cost{T} = d) &= \sum_{i=1}^{n} \leaves{T}{d}{i} 2^{-d}.
\label{eq:ddg-tree-cost-dist}
\end{align}
The \textit{entropy cost} is the expected entropy consumption:
\begin{align}
\expect{\cost{T}} &= \sum_{d=0}^{\infty} d \sum_{i=1}^{n} \leaves{T}{d}{i} 2^{-d}.
\label{eq:ddg-tree-cost-avg}
\end{align}
The \textit{entropy toll}
\begin{align}
\toll{T} \defeq \expect{\cost{T}} - H(\bp_T) \geq 0
\label{eq:ddg-tree-toll}
\end{align}
of $T$ is the difference between the entropy cost of
$T$ and the Shannon entropy of its output distribution $\bp_T$.
The Shannon entropy is a tight lower bound on the entropy cost by
Shannon's source coding theorem~\citep{shannon1948}.
\Cref{table:notation} summarizes these notations and other
symbols used throughout the paper.

\begin{table*}[t]
\caption{Overview of notation.}
\label{table:notation}
\begin{adjustbox}{max width=\linewidth}
\begin{tabular}{|llll|}
\hline
\textbf{Symbol}
& \textbf{Description}
& \textbf{Definition}
& \textbf{Reference}
\\\hline\hline
$(a_1, \dots, a_n)$
& coprime integer weights of target distribution
& $a_i \geq 1$; $\gcd(a_1, \dots, a_n) = 1$
& \Cref{problem:sampling}
\\
$m$
& sum of integer weights of target distribution
& $m=a_1+\dots+a_n$
& \Cref{problem:sampling}
\\
$p_i$
& target probability of outcome $i \in \set{1, \dots, n}$
& $p_i \defeq a_i/m$
& \Cref{problem:sampling} \rule{0pt}{2ex}
\\
$\bp$
& target probability distribution
& $\bp \defeq (p_1,\dots,p_n)$
& \Cref{problem:sampling}
\\\hline
$\Ho{\cdot}$
& weighted information content function
& $\Ho{x} \defeq x\log(1/x)$
& \Cref{theorem:ky-toll} \rule{0pt}{2ex}
\\
$\Hb{\cdot}$
& binary entropy function
& $\Hb{x} \defeq \Ho{x} + \Ho{1-x}$
& \Cref{theorem:topsoe}
\\
$H(\cdot)$
& Shannon entropy function
& $H(\bp) \defeq \sum_{i=1}^n \Ho{p_i}$
& Page~\pageref{sec:introduction}
\\ \hline
$\set{0,1}^*$
& set of all finite-length binary strings
& $\set{0,1}^* \defeq \bigcup_{n=0}^\infty \set{0,1}^n$
& \Cref{sec:preliminaries-ddg} \rule{0pt}{2.5ex}
\\
$T$
& discrete distribution generating (DDG) tree
& $T \in (\set{0,1}^*\rightharpoonup\nat)$
& \Cref{sec:preliminaries-ddg}
\\
$\leaves{T}{d}{i}$
& number of leaves in $T$ at depth $d$ with label $i$
& ~
& Page~\pageref{eq:ddg-tree-probs}
\\
$\bp_T$
& output distribution of $T$
& $p_{T,i} \defeq \sum_{d=0}^\infty \ell_T(d,i)2^{-d}$
& \Cref{eq:ddg-tree-probs}
\\
$\cost{T}$
& entropy consumption of $T$
& $\Pr(\cost{T} = d) \defeq \sum_{i=1}^{n} \leaves{T}{d}{i} 2^{-d}$
& \Cref{eq:ddg-tree-cost-dist}
\\
$\toll{T}$
& entropy toll of $T$
& $\toll{T} \defeq \expect{\cost{T}} - H(\bp_T)$
& \Cref{eq:ddg-tree-toll}
\\
$\epsd{\cdot}$
& $d$\textsuperscript{th} bit in binary expansion
& $\epsd{x} \defeq \floor{2^d x} \bmod 2$ for $x \in \real$
& \Cref{theorem:knuth-yao}
\\
$\nuu{\cdot}$
& ``new'' entropy function
& $\nuu{x} \defeq \sum_{d=0}^{\infty}d\epsd{x}2^{-d}$
& \Cref{theorem:knuth-yao}
\\
~
& ~
& $\nuu{g} \defeq \sum_{d=D}^{\infty} d g_d z^d$ for $g \in \real((z))$
& \Cref{definition:nu-entropy}
\\
$\trel{\cdot}$
& relative entropy toll contribution
& $\trel{x} \defeq (\nu(x)-\Ho{x}) / x$
& \Cref{definition:relative-toll}
\\
$\ky[\bp]$
& DDG tree of \citeauthor{knuth1976} entropy-optimal sampler
& ~
& \Cref{theorem:knuth-yao}
\\ \hline
$(A_1, \dots, A_n)$
& integer weights of target distribution
& $A_i = c a_i$
& \Cref{sec:fldr} \rule{0pt}{2ex}
\\
$\fldr[A_1,\dots,A_n]$
& DDG tree of Fast Loaded Dice Roller sampler
& \Cref{alg:fldr}
& \Cref{sec:fldr}
\\
$\fldr[\bp]$
& minimum-depth $\fldr$ DDG tree for $\bp$
& $\fldr[\bp] \equiv \fldr[a_1,\dots,a_n]$
& \Cref{remark:fldr-scaling}
\\
$k$
& depth of $\fldr[\bp]$ tree
& $k \defeq \ceil{\log(m)}$ (i.e., $2^{k-1} < m \leq 2^k$)
& \Cref{problem:sampling}
\\
$a_0$
& integer weight of reject outcome for $\fldr[\bp]$
& $a_0 \defeq 2^k-m$
& \Cref{eq:proposal-distribution}
\\
$M$
& sum of integer weights
& $M=cm=A_1+\dots+A_n$
& \Cref{sec:fldr}
\\
$K$
& depth of $\fldr[A_1,\dots,A_n]$ tree
& $K = \ceil{\log(M)}$ (i.e., $2^{K-1} < M \leq 2^K$)
& \Cref{problem:sampling}
\\
$A_0$
& weight of reject outcome for $\fldr[A_1,\dots,A_n]$
& $A_0 \defeq 2^K-M$
& \Cref{eq:amplified-proposal-distribution}
\\
$q_i$
& proposal probability of outcome $i \in \set{0, \dots, n}$
& $q_i \defeq A_i/2^K$
& \Cref{eq:proposal-distribution}
\\
$\bq$
& proposal distribution for $\fldr[A_1,\dots,A_n]$
& $\bq \defeq (q_0,\dots,q_n)$
& \Cref{eq:proposal-distribution}
\\
$\trelfldr{A_i, M}$
& relative $\fldr$ toll contribution
& $\begin{aligned}[t]
  &\trelfldr*{A_i, M} \defeq
  \trel*{A_i / 2^K}\\[-5pt]
  &+ (2^K/M)
    \left( \Ho*{M/2^K} + \nu(1 - M/2^K) \right)
  \end{aligned}$
& \Cref{definition:fldr-toll}
\\ \hline
$\aldr[\bp,K]$
& DDG tree of Amplified Loaded Dice Roller sampler
& \cref{alg:aldr}
& \Cref{sec:aldr-main-idea} \rule{0pt}{2ex}
\\
$K$
& amplification depth parameter for $\aldr[\bp,K]$
& $K \geq k$
& \Cref{sec:aldr-main-idea}
\\
$c_K$
& amplification factor for $\aldr[\bp, K]$
& $c \equiv c_K \defeq \floor{2^K/m}$
& \Cref{sec:aldr-main-idea}
\\ \hline
$\real[[z]]$
& ring of formal power series
& $\real[[z]] \defeq \set*{\sum_{d=0}^{\infty} g_d z^d \mid g_0,g_1,\ldots \in \real}$
& \Cref{theorem:generating-fn} \rule{0pt}{2.5ex}
\\
$\real((z))$
& ring of formal Laurent series
& $\real((z)) \defeq \set*{\sum_{d=D}^{\infty} g_d z^d \mid D \in \integers, g_d \in \real}$
& \Cref{definition:nu-entropy}
\\
$z \mapsto 1/2$
& series evaluation
& $\left[ \sum_{d=D}^{\infty} g_d z^d \right]_{z \mapsto 1/2} \defeq \sum_{d=D}^{\infty} g_d 2^{-d}$
& \Cref{lemma:nu-derivation}
\\
$\log(\cdot)$
& base $2$ logarithm
& $\log(x) \defeq \ln(x) / \ln(2)$
& ~
\\ \hline\hline
\end{tabular}
\end{adjustbox}
\end{table*}

\subsection{Entropy-Optimal DDG Trees}

\Citet{knuth1976} settle the problem of constructing a DDG tree sampler for
any distribution $\bp$ whose entropy cost is minimal among the
class of all DDG trees with output distribution $\bp$.

\begin{theorem}[{\Citet[Theorem 2.1]{knuth1976}}]
\label{theorem:knuth-yao}
Let $\bp \defeq (p_1, \dots, p_n)$ denote a discrete probability
distribution over $n$ outcomes.
Let
$\mathcal{T}(\bp) = \set{T : \set{0,1}^{*} \rightharpoonup \nat \mid \bp_T = \bp}$
denote the set of all DDG trees whose output distribution is $\bp$.
The following statements regarding a DDG tree $T \in \mathcal{T}(\bp)$ are equivalent.
\begin{enumerate}[label={\labelcref{theorem:knuth-yao}.\arabic*)}]
\item For all $d \geq 0$, $T$ minimizes the probability of consuming more than
$d$ coin flips, in the sense that
for all $T' \in \mathcal{T}(\bp)$
\begin{align*}
\Pr(\cost{T} > d) \leq \Pr(\cost{T'} > d).
\end{align*}

\item For each label $i=1,\dots,n$ and depth $d \geq 0$,
the number of leaf nodes with label $i$ at depth $d$ of $T$ satisfies
$\leaves{T}{d}{i} = \epsd{p_i} \in \set{0,1}$, where
$\epsd{x} \defeq \floor{2^d x} \bmod 2$ is the $d$th bit
in the (concise) binary expansion of $x \in [0,1]$.

\item The entropy cost of $T$ is minimal:
\begin{align*}
\expect{\cost{T}}
&= \min\set{ \expect{\cost{T'}} \mid T' \in \mathcal{T}(\bp)} \\
&= \nu(p_1) + \dots + \nu(p_n),
\end{align*}
where $\nu(x) \defeq \sum_{d=0}^{\infty} d \epsd{x}2^{-d}$ is the ``new''
entropy function.
\qedhere
\end{enumerate}
\end{theorem}

Following this theorem, the entropy-optimal DDG tree
can be constructed from the binary expansions of the probabilities $p_i$.
\Citeauthor{knuth1976} also characterize the entropy cost of
entropy-optimal DDG trees.

\begin{theorem}[{\Citet[Theorem 2.2 and Corollary]{knuth1976}}]
\label{theorem:ky-toll}
The entropy cost $\expect{\cost{T}}$ of any entropy-optimal DDG
tree $T$ for $\bp$ satisfies
$H(\bp) \leq \expect{\cost{T}} < H(\bp) + 2$;
and these bounds are tight.
\end{theorem}

\begin{proof}[Proof (Sketch)]
\Cref{eq:ddg-tree-toll} and \cref{theorem:knuth-yao} give a simple
expression for the entropy toll of the entropy-optimal DDG tree $T$
with output distribution $\bp_T = \bp$:
\begin{align*}
\toll{T}
  \defeq \expect{\cost{T}} - H(\bp)
  = \sum_{i=1}^{n}p_i \left[{\frac{\nu(p_i) - \Ho{p_i}}{p_i}}\right],
\end{align*}
where $\Ho{x} \defeq x \log(1/x)$, and by continuity $\Ho{0} \defeq 0$
(the ${}_1$ subscript indicates that the function is applied to a
single value, rather than a distribution.)
\Citet[Theorem 2.2]{knuth1976} show that
$0 \leq ({\nu(x) - \Ho{x}})/{x} < 2$
for any $x>0$ (cf.~\cref{corollary:relative-toll-bound-constant}),
which implies that the entropy toll satisfies
$0 \leq \toll{T} < 2$.
\end{proof}

\begin{remark}
\label{remark:toll-of-dist}
The notation $\ky[\bp]$
denotes any entropy-optimal DDG tree with output distribution $\bp$.
Every DDG tree $T$ satisfies $\toll{\ky[\bp_T]} \leq \toll{T}$, with
equality if and only if $T$ is entropy optimal.
The gap $\toll{T} - \toll{\ky[\bp_T]} \geq 0$ characterizes the
entropy inefficiency of $T$ relative to the entropy-optimal sampler.
\end{remark}

\paragraph*{Exponential Space Complexity}
\Citet[Theorem 3.5]{saad2020popl} show that
any entropy-optimal DDG tree $T$ for $\bp$
has a finite representation with at most $m$ levels, and that this
bound is tight \citep[Theorem 3.6]{saad2020popl} for infinitely many
target distributions $\bp$ \citep[Remark 3.7]{saad2020fldr} (assuming Artin's
conjecture~\citep{hooley1967} on primitive roots).
Because  $n\log(m)$ bits are needed to encode the $m$-type distribution $\bp$,
a DDG tree with $m$ levels is exponentially large in the input size.
For example, any finite representation of the entropy-optimal DDG tree for
the $\mathrm{Binomial}(50, 61/500)$ distribution has roughly
$5.6\times{10^{104}}$ levels.%
\footnote{The $\mathrm{Binomial}(50, 61/500)$ distribution has denominator
$m = 500^{50} = 2^{100} \cdot 5^{150}$.
Let $\ell=o_{5^{150}}(2)$ be the order of $2$ modulo $5^{150}$, i.e., the smallest
positive integer such that $2^\ell \equiv 1 \pmod{5^{150}}$.
Then, the binary expansions of the probabilities do not repeat until
$100 + \ell$ bits after the binary point,
so a minimal explicit DDG tree for the distribution has $101+\ell$ levels.
Because $o_5(2) = 4$, and $5$ appears only once in the prime factorization of
$2^4 - 1 = 15$, the lifting-the-exponent lemma shows that
$o_{5^r}(2) = 4 \cdot 5^{r-1}$ for all $r \geq 1$.
In particular, $\ell = o_{5^{150}}(2) = 4 \cdot 5^{149}$.
Therefore, the number of levels is $101+4 \cdot 5^{149} \approx 5.6 \times 10^{104}$.}
More generally, a 64-bit machine can natively represent $m$-type
distributions where $m \approx 2^{64}$, highlighting the enormous resources
in practice required by any algorithm whose space complexity is
order $m$.

\subsection{Rejection Sampling}
\label{sec:rejection-sampling}

The rejection method of \citet{neumann1951}
produces exact samples from
$\bp \defeq (p_1, \dots, p_n)$
by using a proposal distribution $\bq \defeq (q_0, q_1, \dots, q_n)$
over a larger domain,
for which there exists a finite bound $B \geq 1$ that satisfies
$p_i \leq B q_i$ for $i=0,\dots,n$
(with $p_0 \defeq 0$).
A sample $I \sim \bq$ is first generated from $\bq$ and then accepted with
probability $\alpha(I) \defeq p_I/(Bq_I)$;
otherwise it is rejected and the process repeats.
The probability of accepting in any given trial is $1/B$.
The number of trials until acceptance follows a geometric
distribution with parameter $1/B$.

\paragraph*{Entropy Cost}

Let $T_\bq$ be any DDG tree with output distribution $\bq$ and let $T_i$ be
entropy-optimal DDG trees for $\mathrm{Bernoulli}(\alpha(i))$, i.e.,
$p_{T_i} = (1-\alpha(i), \alpha(i))$ for $i=0,\dots,n$.
The overall DDG tree $R$ of the rejection sampler has the same structure as
$T_\bq$, where each leaf labeled $i$ in $T_\bq$ is replaced with a new subtree
$T'_i$ in $R$ for the accept-reject decision.
Each tree $T'_i$ is derived from $T_i$, where leaves with label $0$
in $T_i$ are replaced in $T'_i$ with a back edge to the root of $R$ (reject);
and leaves with label $1$ in $T_i$ are relabeled to $i$ in $T'_i$ (accept).
By memorylessness,
the overall entropy cost of $R$ is the expected number of trials times
the per-trial entropy cost:
\begin{align}
\expect{\cost{R}} &= B
  \left(
    \expect{\cost{T_\bq}}
    + \sum_{i=0}^{n}\expect{\cost{T_i}}q_i
  \right),
  \label{eq:expected-cost-rejection}
\end{align}
where $\expect{\cost{T_i}} \leq 2$, as analyzed in \cref{remark:bernoulli-nu}.

\paragraph*{Choice of Proposal}
In light of \cref{eq:expected-cost-rejection},
we argue that, when rejection sampling in the random bit model, it is
natural to consider proposal distributions $\bq$ that have
the following two properties:
\begin{itemize}
\item The probabilities in $\bq$ are dyadic rationals.
\item The acceptance probabilities satisfy $\alpha(i) = 1$ for each $i=1,\dots,n$.
\end{itemize}

The first property is designed to keep $\expect{\cost{T_\bq}}$
in \cref{eq:expected-cost-rejection} small:
if each probability $q_i$ is dyadic with
denominator $2^K = \textrm{poly}(m)$, then we can construct an entropy-optimal
DDG tree $T_\bq \equiv \ky[\bq]$ for $\bq$
whose size is polynomial in the input size $\bp$.
Using nondyadic probabilities, on the other hand, provides no such guarantee.

The second property ensures that $\expect{\cost{T_i}}$ in
\cref{eq:expected-cost-rejection} is zero for each $i=0,\dots,n$, since
each probability $\alpha(i) \in \set{0,1}$ defines a
deterministic accept-reject decision.

The rejection samplers developed in this work will
use proposal distributions $\bq$ that exhibit
these two properties, where we will
seek to minimize the rejection bound $B$ to ensure the overall
entropy cost in \cref{eq:expected-cost-rejection} is small.

\section{Fast Loaded Dice Roller}
\label{sec:fldr}

\begin{listing*}[t]
\setlength{\abovedisplayskip}{0pt}
\setlength{\belowdisplayskip}{0pt}
\captionsetup{hypcap=false}
\begin{minipage}[t]{.495\linewidth}
\begin{algorithm}[H]
\caption{Fast Loaded Dice Roller (Sketch)}
\label{alg:fldr}
\begin{algorithmic}[1]
\Require{List $(A_1, \dots, A_n)$ of positive integers}
\Ensure{Random integer $i$ with probability
  $$p_i \defeq A_i/(A_1+\ldots+A_n)\;\; (1 \leq i \leq n)$$}

  \State Let $M \defeq A_1 + \dots + A_n$

  \State Let $K \defeq \ceil{\log(M)}$

  \State Define the proposal distribution
    \label{item:phorometer}
    \begin{align*}
    \bq \defeq ((2^K - M)/2^K, A_1/2^K, \dots, A_n/2^K)
    \end{align*}

  \State Generate $i \sim \bq$, using the
    entropy-optimal sampler $\ky[\bq]$
    described in \cref{theorem:knuth-yao}
    \label{item:Larix}

  \State If $i = 0$ then go to \cref{item:Larix}; else return $i$
    \label{item:blastosphere}
\end{algorithmic}
\end{algorithm}
\end{minipage}\hfill
\begin{minipage}[t]{.495\linewidth}
\begin{algorithm}[H]
\caption{Amplified Loaded Dice Roller (Sketch)}
\label{alg:aldr}
\begin{algorithmic}[1]
\Require{List $(a_1, \dots, a_n)$ of coprime positive integers;
Amplification rule $r: k \mapsto K$, e.g., $r(k)=2k$
}
\Ensure{Random integer $i$ with probability
  $$p_i \defeq a_i/(a_1+\ldots+a_n)\;\; (1 \leq i \leq n)$$}
\State Let $m \defeq a_1 + \dots + a_n$
\State Let $k \defeq \ceil{\log(m)}$
\State Let $K \defeq r(k)$
\State Let $c \defeq \floor{2^K/m}$
\State Let $A_i \defeq ca_i$, $i=1,\dots,n$
\State Call $\fldr$ (\cref{alg:fldr}) with $(A_1, \dots, A_n)$
\end{algorithmic}
\end{algorithm}
\end{minipage}

\bigskip


\end{listing*}

The Fast Loaded Dice Roller ($\fldr$) \citep{saad2020fldr} is an efficient
rejection sampling algorithm based on the aforementioned idea, whose memory
scales linearithmically with the input size.
Recalling \cref{problem:sampling}, we momentarily assume that $\bp$
is specified by an arbitrary list
$(A_1,\dots,A_n)$ of $n$ positive integers that need not be coprime,
with sum $M$.
$\fldr$ operates as follows:
let $K \defeq \ceil{\log(M)}$, i.e., $2^{K-1} < M \leq 2^K$
and let
$\bq \defeq (q_0, q_1, \dots, q_{n})$ be a proposal
distribution
over $n+1$ outcomes, whose probabilities are dyadic with
\begin{align}
A_{0} \defeq 2^K-M,
&&
q_{i} \defeq A_i / 2^K
&& ( 0 \le i \le n).
\label{eq:proposal-distribution}
\end{align}
The DDG tree of the FLDR sampler---denoted $\fldr[A_1,\dots,A_n]$---is
identical to the entropy-optimal DDG tree $\ky[\bq]$,
except that every leaf with the reject label $0$
becomes a back edge to the root.
Because $\bq$ is dyadic, the tree $\ky[\bq]$ has $K$ levels, avoiding
exponential growth of the depth with the size of $\bp$.

In terms of rejection sampling, the tightest rejection bound
is $B \defeq 2^K/M$; i.e., it is the smallest number that satisfies
$p_i \leq B q_i$ for $i=0,1,\dots,n$, because
$p_i = A_i/M = 2^K/M \cdot A_i/2^K = Bq_i$
($i=1,\dots,n$) and $p_{0} = 0 \leq Bq_{0}$.
It follows that $i \sim \bq$ is accepted with probability $p_i/(Bq_i) = 1$
if $i \in \set{1,\dots,n}$, and $i=0$ is rejected.
This property means that the Bernoulli subtree cost is $\expect{\cost{T_i}} = 0$
in \cref{eq:expected-cost-rejection}, with all entropy used only to sample the
proposal, and in particular,
\begin{equation}
\expect{\cost{\fldr[A_1,\dots,A_n]}} = \frac{2^K}{M}\,\expect{\cost{\ky[\bq]}}.
\label{eq:fldr-cost}
\end{equation}
\Cref{alg:fldr} gives an overview of the $\fldr$ sampler.
\Cref{fig:ddg-ky-fldr} shows examples of the underlying DDG trees.

\begin{figure}[t]
\tikzset{sibling distance=2pt, level distance=10pt}
\centering
\begin{subfigure}[t]{.33\linewidth}
\centering
\begin{adjustbox}{valign=t}
\begin{tikzpicture}
\Tree[.\node[branch,name=root]{}; [ [ [ \node[leaf,name=back]{\phantom{0}}; 1 ] 1 ] 4 ] 4 ]
\draw[red,-latex,out=110,in=180] (back.center) to (root.south);
\end{tikzpicture}
\end{adjustbox}
\end{subfigure}%
\begin{subfigure}[t]{.33\linewidth}
\centering
\begin{adjustbox}{valign=t}
\begin{tikzpicture}
\Tree[.\node[branch,name=root]{}; [ [ \textcolor{red}{R} 1 ] \textcolor{red}{R} ] 4 ]
\end{tikzpicture}
\end{adjustbox}
\end{subfigure}%
\begin{subfigure}[t]{.25\linewidth}
\centering
\begin{adjustbox}{valign=t}
\begin{tikzpicture}
\Tree[.\node[branch,name=root]{}; [ [  \node[leaf,name=b2]{\phantom{0}}; 1 ] \node[leaf,name=b1]{\phantom{0}}; ] 4 ]
\draw[red,-latex] (b1.center) to[out=150,in=180] (root.south);
\draw[red,-latex] (b2.center) to[out=110,in=180] (root.south);
\end{tikzpicture}
\end{adjustbox}
\end{subfigure}

\begin{subfigure}[t]{.33\linewidth}
\caption{\begin{tabular}[t]{@{}c@{}}$\ky[\bp]$\end{tabular}}
\end{subfigure}%
\begin{subfigure}[t]{.33\linewidth}
\caption{\begin{tabular}[t]{@{}c@{}}$\ky[\bq]$\end{tabular}}
\end{subfigure}
\begin{subfigure}[t]{.25\linewidth}
\caption{\begin{tabular}[t]{@{}c@{}}$\fldr[\bp]$\end{tabular}}
\end{subfigure}
\caption{Comparison of DDG trees for $\bp = (1/5, 4/5)$
with $\fldr$ proposal $\bq = (3/8, 1/8, 4/8)$
}
\label{fig:ddg-ky-fldr}
\end{figure}

\begin{remark}
\label{remark:fldr-scaling}
The explicit notation $\fldr[A_1,\dots,A_n]$ is necessary because the space
and entropy costs of \cref{alg:fldr} depend on the scaling of
the inputs $(a_1,\dots,a_n)$ that define the $m$-type distribution
$\bp$ (i.e., these integers need not be coprime).
The notation $\fldr[\bp]$ is defined as $\fldr[a_1,\dots,a_n]$, where
$(a_1,\dots,a_n)$ is the unique list of coprime positive integers that define
$\bp = (a_1/m, \dots, a_n/m)$ and
$m \defeq a_1 + \dots + a_n$ is the smallest integer
for which $\bp$ is $m$-type.
Henceforth, the ``entropy cost (or depth or toll) of $\fldr$'' shall
refer to that of $\fldr[\bp]$ unless otherwise noted.
\end{remark}

\paragraph*{FLDR Space and Entropy Bound}

\Citet[Theorem 5.1]{saad2020fldr} prove that the tree $\fldr[A_1,\dots,A_n]$ has
at most $2(n+1)\ceil{\log(M)} = O(n\log(M))$ nodes, matching the input size
($n \ceil{\log(M)}$ bits to represent $n$ $\ceil{\log(M)}$-bit integers).
Because each label $i \in \set{0,1,\dots,n}$ at a leaf requires
$\ceil{\log(n+1)}$ bits, the overall space complexity is $n\log(M)\log(n)$,
which is linearithmic in the input size.
In addition, \citeauthor{saad2020fldr} prove the following bound on the
entropy cost.

\begin{theorem}[{\Citet[Theorem 5.1]{saad2020fldr}}]
\label{theorem:fldr-bound}
The toll of the FLDR sampler for any distribution $\bp$ is less than $6$.
That is, the entropy cost of FLDR satisfies
$H(\bp) \leq \expect{\cost{\fldr[A_1,\dots,A_n]}} < H(\bp) + 6$.
\end{theorem}

\begin{proof}
\Citet{saad2020fldr} prove that the entropy toll is
\begin{align}
\begin{aligned}
&\toll{\fldr[A_1,\dots,A_n]} \\
&= \begin{aligned}[t]
  \log\left(\frac{2^K}{M}\right)
  &+ \frac{2^K - M}{M} \log\left(\frac{2^K}{2^K - M}\right) \\
  &+ \frac{2^K}{M} \toll{\ky[\bq]}
  \end{aligned}
\end{aligned}
\label{eq:fldr-toll-original}
\end{align}
and note that the three summands in \cref{eq:fldr-toll-original} are bounded by
1, 1, and 4, respectively, which establishes \cref{theorem:fldr-bound}.
We present an alternative proof of \cref{eq:fldr-toll-original}
directly in terms of entropy, which will be useful in future sections.
Let $E$ be the event that an ``accept'' outcome
in $\set{1,\dots,n}$ is obtained from a single draw from
the proposal $\bq$, so that
$\Pr(E) = q_{1} + \cdots + q_{n} = M/2^K$.
The conditional entropy of $\bq$ given $E$
is $H(\bq|E) = H(\bp)$, and $H(\bq|E') = 0$.
Using the chain rule for conditional entropy to write $H(\bq)$ in terms of
$H(\bp)$ and the acceptance probability $M/2^K$ gives
\begin{align}
H(\bq)
&= H(\bq|E) \Pr(E) \begin{aligned}[t]&+ H(\bq|E')\Pr(E') \\ &+ H(E) \end{aligned} \notag \\
&= H(\bp)M/2^K + \Hb{M/2^K}.
\label{eq:q-entropy-decomposition}
\end{align}
Combining the toll expressions \cref{eq:fldr-cost,eq:ddg-tree-toll} with the
proposal entropy decomposition \cref{eq:q-entropy-decomposition} gives
\begin{align}
&\toll{\fldr[A_1,\dots,A_n]} \notag \\
&= \left(2^K/M\right) \left[ H(\bq)+\toll{\ky[\bq]} \right] - H(\bp) \notag \\
&= \left(2^K/M\right) \begin{aligned}[t]
  \big[ H(\bp)M/2^K
  &+ \Hb{M/2^K}  \notag \\
  &+\toll{\ky[\bq]} \big] \notag \\
  &- H(\bp)
  \end{aligned} \notag \\
&= \left(2^K/M\right) \left[ \toll{\ky[\bq]} + \Hb{M / 2^K} \right],
\label{eq:fldr-toll-binary-entropy}
\end{align}
and expanding the binary entropy yields \cref{eq:fldr-toll-original}.
\end{proof}

The specific toll depends on the input $(A_1,\dots,A_n)$.
If $M \in \set{2^K, 2^K-1}$, then FLDR is entropy optimal in the sense of
\cref{theorem:knuth-yao}.
If $M = 2^{K-1}+1$, then FLDR
obtains its worst-case rejection probability
$1-1/B = (2^K-M)/2^K = 1/2 - 2^{-K} \approx 1/2$.

\paragraph*{FLDR Entropy Bound is Tight}

\Cref{theorem:fldr-is-tight} extends the result of
\citet{saad2020fldr} by demonstrating a sequence of probability
distributions for which the entropy cost of FLDR approaches
$H(\bp) + 6$ exponentially quickly in the depth $k$,
which proves that the upper bound in \cref{theorem:fldr-bound} is tight.
The following bound on the binary entropy function is useful.

\begin{proposition}[{\Citet[Theorem 1.1]{topsoe2001}}]
\label{theorem:topsoe}
The binary entropy $\Hb{p} \defeq p \log(1/p) + (1-p) \log(1/(1-p))$
satisfies, for each $p \in [0,1]$,
\begin{align*}
\ln(2) \log(p) \log(1-p) &\leq \Hb{p} \leq \log(p) \log(1-p).
\end{align*}
When $p = 0$ or $p = 1$, these expressions are interpreted as limits,
each with value zero.
\end{proposition}

\begin{theorem}[Tightness of FLDR toll bound]
\label{theorem:fldr-is-tight}
There exists a sequence of rational discrete distributions
$\bp_2, \bp_3, \dots$ whose entropy tolls satisfy
$6 - \toll{\fldr[\bp_k]} = O(k2^{-k})$, for each FLDR tree depth $k \geq 2$.
\end{theorem}

\begin{proof}
\begin{figure}
\tikzset{sibling distance=1pt, level distance=10pt}
\centering
\begin{subfigure}[t]{.5\linewidth}
\centering
\begin{adjustbox}{valign=t}
\begin{tikzpicture}[remember picture]
\Tree[.\node[branch,name=root-ft4]{}; [ \node[leaf,name=back4-ft4]{\phantom{0}}; [ \node[leaf,name=back8-ft4]{\phantom{0}}; [ \node[leaf,name=back16-ft4]{\phantom{0}}; 7 ] ] ] [ 7 [ 7 2 ] ] ]
\draw[red,-latex,out=110,in=180] (back16-ft4.center) to (root-ft4.south);
\draw[red,-latex,out=110,in=180] (back8-ft4.center) to (root-ft4.south);
\draw[red,-latex,out=110,in=180] (back4-ft4.center) to (root-ft4.south);
\end{tikzpicture}
\end{adjustbox}
\end{subfigure}\hfill
\begin{subfigure}[t]{.5\linewidth}
\centering
\begin{adjustbox}{valign=t}
\begin{tikzpicture}[remember picture]
\Tree[.\node[branch,name=root-ft5]{}; [ \node[leaf,name=back4-ft5]{\phantom{0}}; [ \node[leaf,name=back8-ft5]{\phantom{0}}; [ \node[leaf,name=back16-ft5]{\phantom{0}}; [ \node[leaf,name=back32-ft5]{\phantom{0}}; 15 ] ] ] ] [ 15 [ 15 [ 15 2 ] ] ] ]
\draw[red,-latex,out=110,in=180] (back32-ft5.center) to (root-ft5.south);
\draw[red,-latex,out=110,in=180] (back16-ft5.center) to (root-ft5.south);
\draw[red,-latex,out=110,in=180] (back8-ft5.center) to (root-ft5.south);
\draw[red,-latex,out=110,in=180] (back4-ft5.center) to (root-ft5.south);
\end{tikzpicture}
\end{adjustbox}
\end{subfigure}
\begin{subfigure}[t]{.5\linewidth}
\caption{$\toll{\fldr[2,7]} \approx 3.90$}
\end{subfigure}%
\begin{subfigure}[t]{.5\linewidth}
\caption{$\toll{\fldr[2,15]} \approx 4.77$}
\end{subfigure}%

\bigskip

\begin{subfigure}[t]{\linewidth}
\centering
\begin{tikzpicture}[remember picture]
\Tree[.\node[branch,name=root-ft6]{}; [ \node[leaf,name=back4-ft6]{\phantom{0}}; [ \node[leaf,name=back8-ft6]{\phantom{0}}; [ \node[leaf,name=back16-ft6]{\phantom{0}}; [ \node[leaf,name=back32-ft6]{\phantom{0}}; [ \node[leaf,name=back64-ft6]{\phantom{0}}; 31 ] ] ] ] ] [ 31 [ 31 [ 31 [ 31 2 ] ] ] ] ]
\draw[red,-latex,out=110,in=180] (back64-ft6.center) to (root-ft6.south);
\draw[red,-latex,out=110,in=180] (back32-ft6.center) to (root-ft6.south);
\draw[red,-latex,out=110,in=180] (back16-ft6.center) to (root-ft6.south);
\draw[red,-latex,out=110,in=180] (back8-ft6.center) to (root-ft6.south);
\draw[red,-latex,out=110,in=180] (back4-ft6.center) to (root-ft6.south);
\end{tikzpicture}
\caption{$\toll{\fldr[2,31]} \approx 5.31$}
\end{subfigure}
\caption{$\fldr$ trees with tolls rapidly approaching $6$ bits,
as constructed in the proof of \cref{theorem:fldr-is-tight}.}
\label{fig:fldr-tight}
\end{figure}

For $k \ge 2$, set $\bp_k \defeq \left((2^{k-1} - 1)/m, 2/m\right)$ with
$m \defeq \left(2^{k-1} + 1\right)$.
Combining $\log(1+x) \leq x/\ln(2)$ with \cref{theorem:topsoe} yields
a bound on the entropy:
\begin{align*}
H(\bp_k)
&= \Hb{\frac{2}{2^{k-1} + 1}} \\
&\leq \log\left(\frac{2^{k-1} + 1}{2}\right) \log\left(\frac{2^{k-1} + 1}{2^{k-1} - 1}\right) \\
&< (k-1) \frac{2}{2^{k-1} - 1} \frac{1}{\ln(2)}
< \frac{k-1}{2^{k-4}}.
\end{align*}
\Cref{theorem:knuth-yao} shows that the entropy cost of the
entropy-optimal sampler for the proposal distribution
\begin{align*}
\bq_k \defeq \left(\frac{2^{k-1}-1}{2^k}, \frac{2^{k-1}-1}{2^k}, \frac{2}{2^k}\right)
\end{align*}
is given by
\begin{align*}
\expect{\cost{\ky[\bq_k]}}
&= \nuu{\frac{2}{2^k}} + 2\nuu{\frac{2^{k-1} - 1}{2^k}} \\
&= \frac{k-1}{2^{k-1}} + 2 \sum_{i=2}^{k} \frac{i}{2^i} \\
&= \frac{k-1}{2^{k-1}} + 3 - \frac{k+2}{2^{k-1}} \\
&= 3 \left(\frac{2^{k-1} - 1}{2^{k-1}}\right).
\end{align*}
Applying \cref{eq:fldr-cost} gives
\begin{align*}
  \expect{\cost{\fldr[\bp_k]}}
  &= \frac{2^k}{m} \expect{\cost{\ky[\bq_k]}} \\
  &= 6\left(\frac{2^{k-1} - 1}{2^{k-1} + 1}\right).
\end{align*}
Therefore, the toll is bounded as
\begin{align*}
\toll{\fldr[\bp_k]}
&> 6\left(\frac{2^{k-1} - 1}{2^{k-1} + 1} \right) - \frac{k-1}{2^{k-4}} \\
&> 6 - \frac{12}{2^{k-1}} - \frac{2k-2}{2^{k-3}} \\
&= 6 - \frac{2k + 1}{2^{k-3}},
\end{align*}
which approaches $6$ exponentially quickly as $k$ grows.
The bound in \cref{theorem:fldr-bound} is thus tight (cf.~\cref{fig:fldr-tight}).
\end{proof}

\section{Amplified Loaded Dice Roller}
\label{sec:aldr}

The main contribution of this article is a parameterized
family of rejection samplers called the
\textit{Amplified Loaded Dice Roller} ($\aldr$), which exploits the
sensitivity of $\fldr[A_1,\dots,A_n]$ to the
scaling of the integer weights that define $\bp$.
These weights can be scaled in such a way that the entropy cost becomes
strictly less than $H(\bp)+2$, while maintaining the linearithmic space
complexity of the FLDR algorithm.

\subsection{Main Idea}
\label{sec:aldr-main-idea}

Recalling \cref{problem:sampling}, where $(a_1, \dots, a_n)$
are coprime with sum $m$,
the proposal distribution $\bq \defeq (q_0, \dots, q_n)$
in \cref{eq:proposal-distribution}
used by $\fldr[a_1,\dots,a_n]$ has denominator $2^k$,
where $k \defeq \ceil{\log(m)}$.
This proposal can be generalized to a dyadic proposal
$\bq_K \defeq (q_{K,0}, q_{K,1}, \dots, q_{K,n})$
whose denominator is $2^K$ for some integer $K \geq k$.
To retain the desirable property from $\fldr$
that the new acceptance probabilities are
$p_i/(Bq_{K,i}) \in \set{0,1}$, it is necessary and sufficient
to scale the target weights
$(a_1, \dots, a_n)$ by an integer $c \geq 1$:
\begin{align}
A_0 &\defeq 2^K-cm, \,
A_i \defeq ca_i
&& (i = 1,\dots,n),
\\
q_{K,i} &\defeq A_i / 2^K
&& (i = 0,1,\dots,n).
\label{eq:amplified-proposal-distribution}
\end{align}
\Cref{eq:amplified-proposal-distribution} defines a probability distribution
if and only if $M \defeq cm \leq 2^K$ (i.e., $c \leq 2^K/m$), which gives
the following family of proposals whose entropy-optimal samplers
have depth at most $K$:
\begin{align}
\bq_{K,c} \defeq
  \left(\frac{2^K - cm}{2^K}, \frac{ca_1}{2^K}, \dots, \frac{ca_n}{2^K}\right),
\label{eq:amplified-proposal-distribution-family}
\end{align}
for $c = 1, 2, \dots, \floor{2^K/m}$.
Within this family, setting $c = \floor{2^K/m} \eqdef c_K$ is the
optimal choice for maximizing the acceptance probability $cm/2^K$.
We thus define $\bq_K \defeq \bq_{K,c_K}$ or simply $\bq \defeq \bq_K$ when $K$
is a constant or clear from context.
Invoking \cref{alg:fldr} with the amplified weights $(A_1, \dots, A_n)$
in \cref{eq:amplified-proposal-distribution}
gives the $\aldr$ method.
\Cref{alg:aldr} shows the resulting family of rejection samplers, which
take as input coprime positive integers $(a_1,\dots,a_n)$ and an
amplification rule $r: \mathbb{N} \to \mathbb{N}$ that maps the original
($\fldr$) depth $k$ to an amplified depth $K \geq k$.
The DDG tree of \cref{alg:aldr} is denoted $\aldr[\bp, K]$, where $K$
is an expression that involves $k$, e.g., $K = 2k$, or a constant
$K \geq k$ when there is a fixed $k$ under consideration.
Following \citet[Theorem 5.1]{saad2020fldr},
$\aldr[\bp,K]$ is a depth-$K$ DDG tree
and its total number of nodes is at most $2(n+1)\ceil{\log(M)} = 2(n+1)K$,
which is linear in the input size whenever $K = O(k)$.

\begin{remark}
In contrast to the possible ambiguity of $\fldr[\bp]$ in
\cref{remark:fldr-scaling}, $\aldr[\bp, K]$ is uniquely defined because the
inputs $(a_1,\dots,a_n)$ to \cref{alg:fldr} that define $\bp$ are coprime.
Moreover, $\fldr[\bp] \equiv \aldr[\bp, k]$.
If a target distribution is defined in terms of arbitrary weights (not
necessarily coprime), they can be made coprime by dividing them by their
greatest common divisor (GCD).
The computation of the GCD and of the divisions is linearithmic in the
input size; thus, the asymptotic time complexity of \cref{alg:aldr}
is unaffected by this step.
\end{remark}

\begin{remark}[Choice of Depth]
\label{remark:aldr-choice-of-K}
Because $2^{k-1} < m \leq 2^k$, the amplification constant
$c_K = \floor{2^K/m}$ satisfies $2^{K-k} \leq c_K < 2^{K-k+1}$.
If $c_K = 2^{K-k}$ obtains its lowest possible value,
then for $i=1,\dots,n$,
\begin{align*}
q_{K,i} = A_i/2^K = 2^{K-k}a_i/2^K = a_i/2^k = q_{k,i}
\end{align*}
i.e., the amplified proposal
$\bq_K$ is equivalent to the original proposal $\bq_k$.
To ensure that $\bq_K \neq \bq_k$, the new depth $K$
must be large enough to satisfy $2^{K-k} < c_K$:
\begin{alignat}{2}
~ & ~ & 2^K/2^k &< \floor{2^K/m}
\notag \\
\iff &  & 2^K/2^k + 1 &\leq 2^K/m
\notag \\
\iff &  & K &\geq k + \log\left({m}/{(2^k-m)}\right).
\label{eq:aldr-min-K}
\end{alignat}
If the target distribution $\bp$ satisfies
$2^{k-1} < m < 2/3 \cdot 2^k$, then \cref{eq:aldr-min-K}
holds for any choice of depth $K \geq k+1$.
For general $m \neq 2^k$, \cref{eq:aldr-min-K} holds for all $K \geq 2k$
because $\log(m/(2^k-m)) \leq \log(2^k-1) < k$.

In general, as
$c_K = \floor{2^K/m} = 2 c_{K-1} + \epsd[K]{1/m}$, the proposal
distributions satisfy
\begin{align}
\mbox{$c_K$ is even} & \begin{aligned}[t] & \iff \mbox{$c_K = 2 c_{K-1}$} \\
                                       & \iff \bq_{K} = \bq_{K-1},
                                       \end{aligned}
\label{eq:c-case-even}
\\
\mbox{$c_K$ is odd}  & \begin{aligned}[t] &\iff \mbox{$c_K = 2 c_{K-1} + 1$} \\
                                       &\iff \bq_{K} \neq \bq_{K'}\,
                                       (K'=k,\dots,K-1). \hspace{-1cm}
                                       \end{aligned}
\label{eq:c-case-odd}
\end{align}
For an intuitive understanding of this result, observe that
for fixed $(a_1,\dots,a_n)$, the ALDR proposal distributions
$\bq_K \defeq (1 - c_K m/2^K, c_K a_1/2^K, \dots, c_K a_n/2^K)$
are parameterized by
$c_K/2^K = \floor{2^K/m}/2^K$, which is $1/m$ rounded down
to the nearest multiple of $1/2^K$.
Therefore, the proposal distributions change at precisely the
depths $K$ for which the $K$th bit in the binary expansion of $1/m$ is 1.
\end{remark}

\begin{example}
Suppose $\bp = (4/19, 7/19, 8/19)$,
whose FLDR tree has depth $k = 5$.
\Cref{fig:fldr478-tolls} shows the tolls for $\aldr[\bp,K]$ trees with new
depth $K \in \set{5, \dots, 18}$, which is entropy optimal at $K=18$.
%
\Crefrange{fig:fldr478-trees-d5}{fig:fldr478-trees-d18} illustrate the
general trend that larger depth $K$ reduces the entropy cost, at
the expense of increased space.
\begin{figure}[p]

\begin{minipage}[t][\textheight][t]{\columnwidth}

\begin{subfigure}{\linewidth}
\centering
\includegraphics[width=\linewidth]{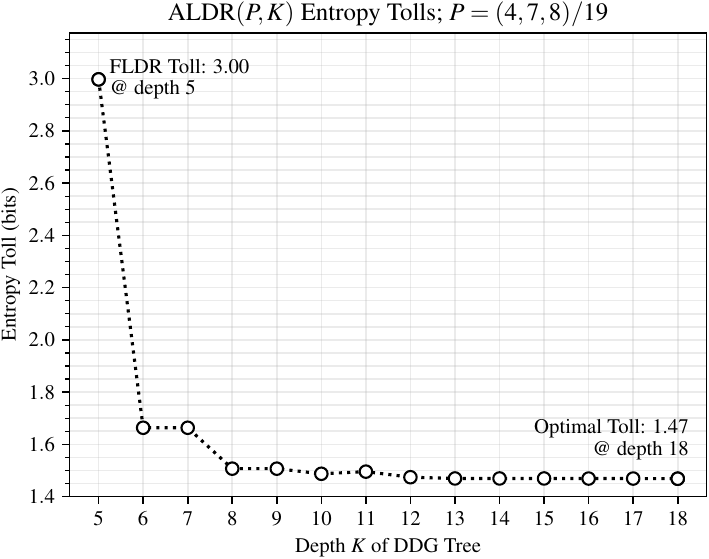}
\caption{Entropy tolls $\tau(\aldr[\bp,K]) = \expect{\cost{\aldr[\bp,K]}} - H(\bp)$ for DDG tree depth $K=5,6,\dots,18$.}
\label{fig:fldr478-tolls}
\end{subfigure}

\vfill

\captionsetup[subfigure]{justification=centering}
\centering
\newcommand{\pz}{\phantom{0}}
\tikzset{sibling distance=1pt, level distance=10pt}
\begin{subfigure}[b]{.4\linewidth}
  \centering
  \begin{adjustbox}{valign=t}
  \begin{tikzpicture}
  \Tree [.\node[branch,name=root]{}; [ [ [ [ \node[leaf,name=b1]{\pz}; 7 ] 7 ] 4 ] [ 7 \node[leaf,name=b2]{\pz}; ] ] [ 8 \node[leaf,name=b3]{\pz}; ] ]
  \draw[red,-latex,out=110,in=180] (b1.center) to (root.south);
  \draw[red,-latex,bend right] (b2.center) to (root.south);
  \draw[red,-latex,out=45,in=0] (b3.center) to (root.south);
  \end{tikzpicture}
  \end{adjustbox}
\end{subfigure}\hfill
\begin{subfigure}[b]{.5\linewidth}
  \centering
  \begin{adjustbox}{valign=t}
  \begin{tikzpicture}
  \Tree [.\node[branch,name=root]{}; [ [ [ [ [ \node[name=b1,leaf]{\pz}; 7 ] \node[name=b2,leaf]{\pz}; ] 4 ] [ 7 \node[name=b3,leaf]{\pz}; ] ] [ 4 8 ] ] [ 7 8 ] ]
  \foreach \i in {1, ..., 3} {
    \draw[red,-latex,out=110,in=180] (b\i.center) to (root.south);
  }
  \end{tikzpicture}
  \end{adjustbox}
\end{subfigure}

\begin{subfigure}[b]{.5\linewidth}
\caption{Depth $K = k = 5$\\ $\bq_5 = (13, 4, 7, 8) / 32$}
\label{fig:fldr478-trees-d5}
\end{subfigure}%
\begin{subfigure}[b]{.5\linewidth}
\caption{Depth $K = 6$ \\ $\bq_6 = (7, 12, 21, 24) / 64$}
\label{fig:fldr478-trees-d6}
\end{subfigure}

\bigskip

\begin{subfigure}{\linewidth}
  \centering
  \begin{tikzpicture}
  \Tree [.\node[branch,name=root]{}; [ [ [ [ [ [ [ \node[leaf,name=b1]{\pz}; 7 ] 7 ] 4 ] 7 ] [ 8 \node[leaf,name=b2]{\pz}; ] ] [ 4 7 ] ] [ 4 8 ] ] [ 7 8 ] ]
  \foreach \i in {1, 2} {
    \draw[red,-latex,out=110,in=180] (b\i.center) to (root.south);
  }
  \end{tikzpicture}
  \caption{Depth $K = 8$ \\ $\bq_8 = (9, 52, 91, 104) / 256$}
  \label{fig:fldr478-trees-d8}
\end{subfigure}

\bigskip

\tikzset{sibling distance=1pt, level distance=6pt}

\begin{subfigure}\linewidth
  \centering
  \begin{tikzpicture}
  \Tree [.\node[branch,name=root]{}; [ [ [ [ [ [ [ [ [ [ [ [ [ [ [ [ [ \node[leaf,name=b1]{\pz}; 7 ] 7 ] 4 ] 8 ] 4 ] 8 ] 7 ] 4 ] 4 ] [ 7 8 ] ] [ 4 8 ] ] [ 4 8 ] ] [ 7 8 ] ] [ 4 7 ] ] [ 7 8 ] ] [ 4 7 ] ] [ 4 8 ] ] [ 7 8 ] ]
  \draw[red,-latex,out=90,in=180] (b1.center) to (root.south);
  \end{tikzpicture}
  \caption{Depth $K = 18$ (Entropy Optimal) \\ $\bq_{18} = (1, 55188, 96579, 110376) / 262144$}
  \label{fig:fldr478-trees-d18}
\end{subfigure}

\vfill

\caption{DDG trees and entropy tolls of $\aldr[\bp,K]$
for $\bp = (4,7,8)/19$ and various DDG tree depths $K$.
$\aldr[\bp,5]$ coincides with $\fldr[4,7,8]$; and
$\aldr[\bp,18]$ coincides with $\ky(\bp)$ from \cref{theorem:knuth-yao}.}
\label{fig:fldr478}
\end{minipage}
\end{figure}

\end{example}

\subsection{Analysis of Entropy Cost}

\begin{proposition}
\label{prop:rejection-probability-decreasing}
The rejection probabilities of the amplified
proposals~\labelcref{eq:amplified-proposal-distribution} are monotonically
decreasing with $K$, i.e., if $\bq_{K+1} \neq \bq_{K}$ then $q_{K+1, 0} < q_{K,0}$.
\end{proposition}

\begin{proof}
If $\bq_{K+1} \neq \bq_K$, then
$c_{K+1} = 2 c_{K} + 1$
from \cref{eq:c-case-odd}, which implies that
\begin{align*}
q_{K+1,0}
\defeq 1 - \frac{m c_{K+1}}{2^{K+1}}
&= 1- \frac{m (2c_K + 1)}{2^{K+1}}\\
&< 1- \frac{m 2c_K}{2^{K+1}} \\
&= 1- \frac{m c_K}{2^{K}}
\eqdef
q_{K,0}.
\qedhere
\end{align*}
\end{proof}

\begin{proposition}
\label{proposition:amplified-reject-prob}
In the amplified proposal $\bq_K$ from \cref{eq:amplified-proposal-distribution},
the reject outcome has probability
\begin{equation*}
q_{K,0} < \frac{1}{2^{K-k} + 1}.
\qedhere
\end{equation*}
\end{proposition}

\begin{proof}
There are two cases on $m/2^k$ to consider:
\begin{enumerate}[label={Case \arabic*.}]
\item
Suppose that
\begin{equation*}
\frac{m}{2^k} > 1 - \frac{1}{2^{K-k} + 1},
\end{equation*}
so $\bq_K = \bq_k$.
The rejection probability is
\begin{equation*}
q_{K,0} = 1 - {m}/{2^k} < {1}/{(2^{K-k} + 1)}.
\end{equation*}
\item
Suppose that
\begin{equation*}
\frac{m}{2^k} < 1 - \frac{1}{2^{K-k} + 1} = \frac{2^{K-k}}{2^{K-k} + 1}.
\end{equation*}
The rejection probability is
\begin{align*}
q_{K,0}
&= \left(2^K - \floor{2^K / m} m\right)/2^K \\
&= \left(2^K \bmod m\right)/2^K \\
&< {m}/{2^K} \\
&= 2^{k-K} \ {m}{/2^k} \\
&< 1/\left(2^{K-k} + 1\right).
\end{align*}
\end{enumerate}

In either case, $q_{K,0} < 1/(2^{K-k} + 1)$, and
this bound is tight along
$m = \floor{2^K/(2^{K-k} + 1)}$ as $k \to \infty$.
\end{proof}

\begin{theorem}[Generic bound on toll of ALDR]
\label{theorem:aldr-generic-bound}
The entropy cost of the $\aldr[\bp, K]$ sampler satisfies
$\expect{\cost{\aldr[\bp, K]}} < H(\bp) + 2 + (4 + K-k)2^{k-K}$.
\end{theorem}

\begin{proof}
\begin{figure}[t]
\centering
\includegraphics[width=\linewidth]{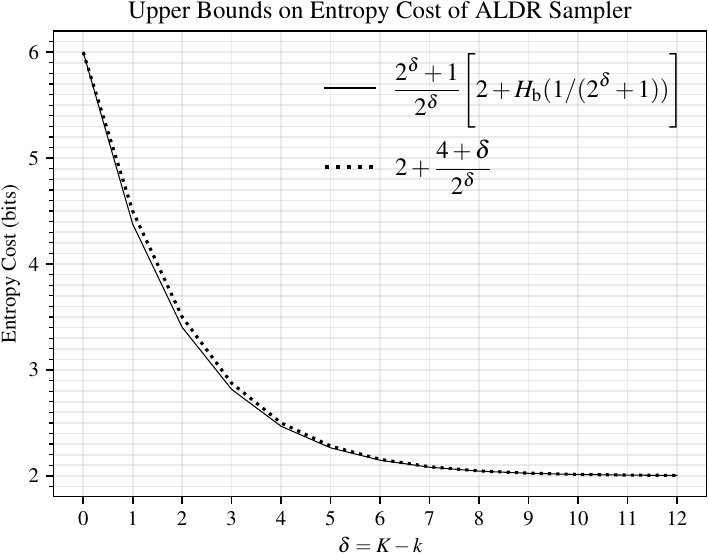}
\caption{Upper bound on the entropy cost of the ALDR sampler with increasing
depth $K$, using the generic bound $\toll{\ky[\bq_K]} < 2$.
The solid line shows a bound in terms of the binary entropy
function \cref{eq:aldr-generic-actual} and the dashed
line shows a simple upper bound on this quantity following \cref{theorem:aldr-generic-bound},
all as a function of $\delta \defeq K-k$.}
\label{fig:aldr-generic-toll-bound}
\end{figure}

If $K=k$, then \cref{theorem:fldr-bound} applies directly.
Otherwise, applying \cref{theorem:ky-toll} and
\cref{proposition:amplified-reject-prob} to bound the toll $\toll{\ky[\bq]}$ and
the rejection bound $2^K/M$ in \cref{eq:fldr-toll-binary-entropy} yields
\begin{align}
&\toll{\aldr[\bp, K]} \notag \\
&< \frac{2^K}{M} \left( 2 + \Hb{M/2^K} \right) \notag \\
&< \frac{2^{K-k}+1}{2^{K-k}} \left[ 2 + \Hb{1/\left(2^{K-k}+1\right)} \right]
  \label{eq:aldr-generic-actual} \\
&= \left(1+\frac{2^k}{2^K}\right) \left(2 + \log\left(1+\frac{2^k}{2^K}\right) \right) + \frac{2^k}{2^K}\log\left(\frac{2^K}{2^k}\right) \notag \\
&= 2 + \frac{2^k}{2^K} \left(2 + K-k \right) + \left(1+\frac{2^k}{2^K}\right)\log\left(1+\frac{2^k}{2^K}\right). \notag
\end{align}
Therefore, it suffices to show that whenever $x$ is positive,
$(1+2^x) \log(1+2^{-x}) \leq 2$.
When $x$ is zero, the value of this function is $(1+2^x) \log(1+2^{-x}) = 2$.
For $x > 0$,
\begin{align*}
\frac{\diff}{\diff{x}} (1+2^x) \log(1+2^{-x}) = 2^x \ln(1+2^{-x}) - 1 < 0,
\end{align*}
which completes the proof (cf.~\cref{fig:aldr-generic-toll-bound}).
\end{proof}

Recall that $\toll{\ky[\bq]} \in [0, 2)$ is the entropy toll from a single
iteration of the FLDR rejection loop, which uses the entropy-optimal sampler for
$\bq$.
The excess $O((K-k)/2^{K-k})$ beyond this toll decays very
quickly in $\delta \defeq K-k$.
\Cref{fig:aldr-generic-toll-bound} shows how these values grow as $\delta$
increases.
This bound is (asymptotically) the best possible when using the
generic result $\toll{\ky[\bq]} \in [0,2)$.
However, this excess decreases so rapidly that it is possible to bound the
entire toll of the ALDR sampler to $2$ bits with linear amplification
$\delta = O(k)$, as demonstrated in \cref{theorem:aldr-2k-toll-two}.

\subsection{Tighter Bound on Entropy Cost for \texorpdfstring{$K = 2k$}{K=2k}}

The main theorem of this section states that $\aldr[\bp,2k]$ achieves a
toll strictly less than $2$ bits with only a doubling of memory relative to
$\fldr[\bp] \equiv \aldr[\bp,k]$.
The result demonstrates that achieving efficient space and runtime
is compatible with achieving an entropy cost in the entropy-optimal range.

\begin{theorem}[Bounding the toll of ALDR]
\label{theorem:aldr-2k-toll-two}
For any $K \geq 2k$, the entropy toll of $\aldr[\bp,K]$ is bounded by two.
That is, the entropy cost of the ALDR sampler satisfies
$H(\bp) \leq \expect{\cost{\aldr[\bp,K]}} < H(\bp) + 2$.
\end{theorem}

\Citet[Theorem 2.2]{knuth1976} prove that every entropy-optimal sampler satisfies
this bound using the fact that
$({\nu(x) - \Ho{x}})/{x} < 2$,
but this alone will not suffice to prove that the toll of ALDR is
strictly less than $2$ (cf.~\cref{theorem:aldr-generic-bound}).
Instead, we need a tighter upper bound parametric in $x$.
The proof of \cref{theorem:aldr-2k-toll-two} requires several intermediate
results.
We first define relative tolls to split up the toll contributions by weight,
then bound these relative tolls, and finally combine these bounds based on the
weights of the input distribution, paying special attention to the case when
$p_i$ is a power of two, which turns out to give the worst-case relative toll
for ALDR.

\begin{definition}
\label{definition:relative-toll}
The \textit{relative toll} of $x \in [0,1]$ is
\begin{equation}
\label{eq:relative-toll}
\trel{x} \defeq \frac{\nu(x) - \Ho{x}}{x},
\end{equation}
where $\Ho{x} \defeq x \log(1/x)$, so $H(\bp) = \sum_{i=1}^{n} \Ho{p_i}$
and $\toll{\ky[\bp]} = \sum_{i=1}^{n} p_i \trel{p_i} = \expect{\trel{p_i}}$.
\end{definition}

\begin{lemma}
\label{lemma:bit-shift-relative-toll}
The relative toll is invariant under bit shifts of its argument.
That is, $\trel{2^a x} = \trel{x}$ for $a \in \integers$.
\end{lemma}

\begin{proof}
Apply the equation
\begin{align*}
\trel{2x}
&\defeq \frac{\nu(2x)-\Ho{2x}}{2x} \\
&= \frac{(2\nu(x)-2x)-(2\Ho{x}-2x)}{2x} \\
&= \frac{\nu(x)-\Ho{x}}{x}
\eqdef \trel{x}
\end{align*}
(repeatedly as necessary for larger powers of two).
\end{proof}

The next results establish useful facts about $\nu$.

\begin{remark}[$\nu$ entropy of Bernoulli distributions]
\label{remark:bernoulli-nu}
Because $\nuu{1-2^{-a}} = 2 - (a+2)2^{-a}$ for any integer $a$,
the entropy cost of the entropy-optimal sampler for a
Bernoulli distribution with dyadic parameter $x = b/2^a$ in lowest terms (i.e., either
$b$ is odd or $x = 0 / 2^0$) is $\nu(x) + \nu(1-x) = \nuu{2^{-a}} + \nuu{1-2^{-a}} = 2 - 2^{1-a}$.
Similarly, $\nu(x) \leq 2 - \nu(1-x)$ for any $x \in [0,1]$,
with strict inequality if and only if $x$ is dyadic.
\end{remark}

\begin{lemma}[Piecewise linear relative toll bound]
\label{lemma:relative-toll-bound-linear}
For any $a \geq 2$ and $b \in \nat$, the relative toll of
any real number $x \in (2^{-b} - 2^{-b-a+1}, 2^{-b} - 2^{-b-a}]$
is less than two:
\begin{equation*}
\trel{x} < 2 - (1 - 2^b x) (a - 2 + 1/\ln(2)).
\qedhere
\end{equation*}
\end{lemma}

\begin{proof}
The relative toll function and corresponding bounds are plotted in
\cref{fig:relative-tolls}.
%
\begin{figure}
\centering
\includegraphics[width=\linewidth]{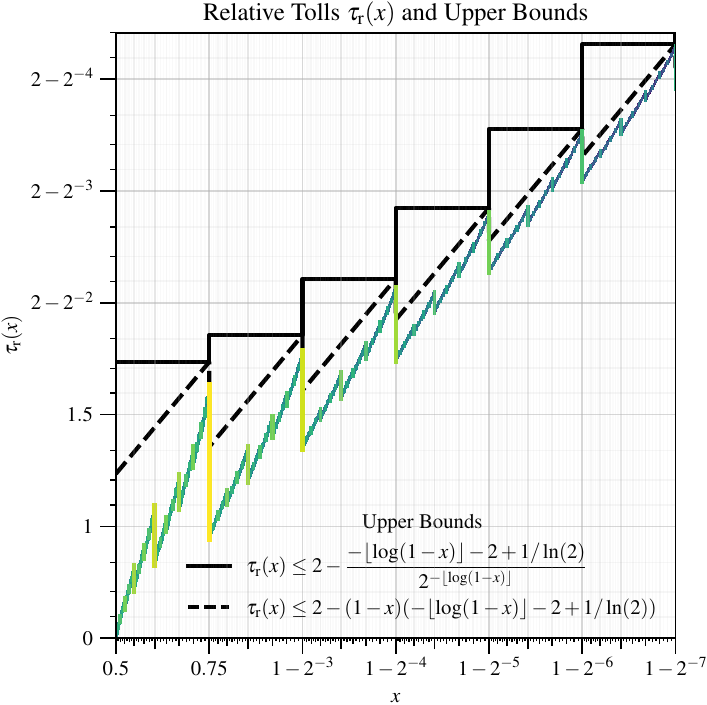}
\caption{Relative tolls and upper bounds for $x \in [1/2,1)$.
At each dyadic rational, a vertical segment connects the
left and right limits of $\trel{x}$.
Dyadic rationals with larger $2$-adic valuation are
plotted with thicker lines.}
\label{fig:relative-tolls}
\end{figure}

First, apply \cref{lemma:bit-shift-relative-toll} to reduce to the case
$x \in (1 - 2^{1-a}, 1 - 2^{-a}]$.
Then $1-x \in [2^{-a}, 2^{1-a})$ so that
\begin{equation*}
\nu(1-x) \geq a 2^{-a} + (1 - 2^{-a} - x)(a+1).
\end{equation*}
\Cref{remark:bernoulli-nu} implies that
\begin{align*}
\trel{x}
&\leq \frac{
  2 - \nu(1-x)
}{x} + \log(x) \\
&\leq \frac{
  2 - a 2^{-a}
  - (1 - 2^{-a} - x)(a+1)
}{x} + \log(x) \\
&< \frac{2-a}{x} + a + \frac{x-1}{\ln(2)} \\
&\leq (2-a)(2-x) + a + \frac{x-1}{\ln(2)} \\
&= 2 - (1-x)(a-2+1/\ln(2)). \tag*{\qedhere}
\end{align*}
\end{proof}

\begin{corollary}[Piecewise constant relative toll bound]
\label{corollary:relative-toll-bound-constant}
For any $a \geq 2$ and $b \in \nat$, the relative toll of
any real $x \in [2^{-b-1}, 2^{-b} - 2^{-b-a}]$
can be bounded strictly less than two as
$\trel{x}
< 2 - 2^{-a}(a-2+1/\ln(2))$.
\end{corollary}

\begin{proof}
Apply \cref{lemma:relative-toll-bound-linear}, using $1-2^b x \geq 2^{-a}$
and the fact that
$2 - 2^{-a}(a-2+1/\ln(2))$ is an increasing function of $a \geq 2$,
as well as $\trel{2^{-b-1}} = 0$.
\end{proof}

\begin{definition}
\label{definition:fldr-toll}
The \textit{relative FLDR toll} of a weight $A$ given weight sum $M$
and $K \defeq \ceil{\log(M)}$ is
\begin{align*}
\begin{aligned}[b]
&\trelfldr{A, M}
\defeq \trel{\frac{A}{2^K}} +
\frac{2^K}{M}
  \left[
        \begin{aligned}
        &\Hb{M/2^K} \\
        &+ \nuu{1 - M/2^K} \\
        &- \Ho{1 - M/2^K}
        \end{aligned}\right]\span \\
&= \trel{\frac{A}{2^K}} +
  \frac{2^K}{M}
  \left( \Ho{\frac{M}{2^K}}
    + \nuu{1 - \frac{M}{2^K}} \right)\!.
\end{aligned}\span \tag*{\qedhere}
\end{align*}
\end{definition}


\begin{remark}
Rewriting \cref{eq:fldr-toll-binary-entropy} as
\begin{align}
&\toll{\fldr[A_1, \ldots, A_n]}
\notag \\
&= \frac{2^K}{M}
  \left( \Hb{M/2^K}
        + \mathbb{E}_{i \sim \bq}\left[\trel{A_i / 2^K}\right] \right)
\notag \\
&= \mathbb{E}_{i \sim \bp}\left[
  \trel{A_i / 2^K} +
  \frac{2^K}{M}
  \left(\begin{aligned}
    &\Hb{M/2^K} \\
    &+ \nuu{A_0 / 2^K} \\
    &- \Ho{A_0 / 2^K}
    \end{aligned}\right)
\right]
\notag \\
&= \mathbb{E}_{i \sim \bp}\left[ \trelfldr{A_i, M} \right]
\label{eq:entropy-difference-convex-combination}
\end{align}
confirms that \cref{definition:fldr-toll}
is consistent with our usage of the relative toll.
\end{remark}

\begin{lemma}
\label{lemma:aldr-pow-two-toll}
When $K \geq 2k$ and $m$ is not a power of two but $a_i / m$ is a power of two,
the relative FLDR toll of $A_i$ given weight sum $M$ is bounded as
$2 \leq \trelfldr{A_i, M} < 2 ( 1 + A_0/2^K )$.
\end{lemma}

\begin{proof}
When $A_i/M$ is a power of two, \cref{lemma:bit-shift-relative-toll} shows that
$\trel{A_i/2^K} = \trel{M/2^K}$.
Then, applying \cref{remark:bernoulli-nu} to \cref{definition:fldr-toll} gives
\begin{align*}
&\trelfldr{A_i, M}\\
&= \trel{M / 2^K} +
\frac{2^K}{M}
\left( \Ho{M/2^K}
  + \nuu{1 - M/2^K} \right)
\\
&= \frac{2^K}{M} \left(
  \nuu{M/2^K} + \nuu{1 - M/2^K}
\right)
\\
&= \frac{2^K}{M} \left(
  2 - 2\gcd(M,2^K)/2^K
\right).
\end{align*}
Therefore, $\trelfldr{A_i, M}$ equals $0$ if $m$ is a power of two
and equals $2$ if $2^K - M$ is a power of two.
Otherwise,
\begin{align*}
&\frac{2^K}{M} \left( 2 - 2\gcd\left(M,2^K\right)/2^K \right) \\
&= 2 \left( 1 + \frac{A_0}{M} - \frac{\gcd\left(A_0, 2^K\right)}{M} \right) \\
&\leq 2 \left( 1 + \frac{A_0-1}{M} \right) \\
&< 2 \left( 1 + \frac{A_0}{2^K} \right),
\end{align*}
because $\left(A_0/M\right)-\left(A_0/2^K\right) = A_0^2/\left(M2^K\right) < 1/M$ whenever $K \geq 2k$.
\end{proof}

\begin{lemma}[Tightness of $q_i \approx p_i$]
\label{lemma:q-p-tightness}
Suppose that $m$ is not a power of two and $p_i \defeq a_i / m$ is not a power of two.
If $K \geq 2k$, then $p_i$ and $q_i \defeq A_i / 2^K$ share the same most-significant bit.
\end{lemma}

\begin{proof}
The following are equivalent:
\begin{inlinelist}[label=(\roman*)]
\item $p_i$ and $q_i$ share the same most significant bit;
\item $\floor{\log(p_i)} = \floor{\log(q_i)}$;
\item $2^{-1-\floor{\log(p_i)}} q_i \in [1/2, 1)$.
\label{item:share-msb}
\end{inlinelist}
We will prove \cref{item:share-msb}.
Consider the function $s(x) \defeq 2^{-1-\floor{\log(x)}} x$, which
bit-shifts any positive real $x$ into the interval $[1/2, 1)$.
Because $a_i / m$ is not a power of two,
\begin{equation*}
s(p_i) = 2^{-1-\floor{\log(a_i/m)}}\frac{a_i}{m} \in (1/2, 1) \cap \frac{1}{m} \integers,
\end{equation*}
so $s(p_i) \geq 1/2 + 1/2m$ (or $s(p_i) \geq 1/2 + 1/m$ if $m$ is even).
The bounds $m \leq 2^k-1$ and $M > 2^K - 2^k$ give
\begin{align*}
&2^{-1-\floor{\log(a_i/m)}}\frac{A_i}{2^K}\\
&\geq \frac{m+1}{2m}\frac{M}{2^K}\\
&> \frac{2^k}{2(2^k-1)}\frac{2^K-2^k}{2^K}
\geq 1/2,
\end{align*}
which proves \cref{item:share-msb}.
\end{proof}

\begin{lemma}[Rejection toll bound]
\label{lemma:aldr-rejection-toll}
The rejection contribution to the relative FLDR toll is bounded in terms of the
rejection weight $A_0 \defeq 2^K-M$ as
\begin{equation*}
\begin{aligned}[b]
&\frac{2^K}{M} \left( \Ho{\frac{M}{2^K}} + \nuu{1 - \frac{M}{2^K}\right)} \\
  &\qquad < \frac{A_0}{M}\left(K+3-\floor{\log(A_0)}\right).
\end{aligned}
\qedhere
\end{equation*}
\end{lemma}

\begin{proof}
The leading bit of $A_0 / 2^K$ is at position $K - \floor{\log(A_0)}$, so
\begin{equation*}
\nuu{A_0 / 2^K} < (K + 1 - \floor{\log(A_0)}) \left(A_0 / 2^K\right).
\end{equation*}
Then
\begin{align*}
&\frac{2^K}{M} \left( \Ho{M/2^K} + \nuu{1 - M/2^K} \right) \\
&< \log\left(2^K/M\right) + \frac{2^K}{M} \frac{A_0}{2^K} (1-\floor{\log\left(A_0/2^K\right)}) \\
&< \frac{A_0}{M \ln(2)} + \frac{A_0}{M}\left(K+1-\floor{\log(A_0)}\right) \\
&< \frac{A_0}{M}\left(K+3-\floor{\log(A_0)}\right). \tag*{\qedhere}
\end{align*}
\end{proof}

These lemmas provide the tools needed to prove \cref{theorem:aldr-2k-toll-two}
(i.e., $\toll{\aldr[\bp, K]} < 2$ for all $\bp$ and all $K \geq 2k$),
given in \cref{appx:proof-theorem-aldr-2k-toll-two} of the online supplement.
The next result shows that
$K=2k$ is the smallest depth satisfying the entropy bound in
\cref{theorem:aldr-2k-toll-two}.

\begin{theorem}[Depth $K\,{=}\,2k$ is minimal]
\label{theorem:aldr-doubling-minimal}
For every $k \geq 4$, there exists a probability distribution $\bp$
such that $\fldr[\bp]$ has depth $k$ and the $\aldr$
entropy tolls satisfy
$\toll{\aldr[\bp, K]} > 2$
for $K = k, \dots, 2k-1$.
\end{theorem}

\begin{proof}
\begin{figure}
\tikzset{sibling distance=3pt, level distance=10pt}
\centering
\begin{subfigure}[t]{.325\linewidth}
\centering
\begin{adjustbox}{valign=t}
\begin{tikzpicture}[remember picture]
\Tree[.\node[branch,name=root-kiswa]{}; [ [ \node[leaf,thick,name=back8-kiswa]{\phantom{0}}; [ 7 \node[leaf,name=back16-kiswa]{\phantom{0}}; ] ] [ 6 7 ] ] [ 6 7 ] ]
\draw[red,-latex,out=110,in=180] (back16-kiswa.center) to (root-kiswa.south);
\draw[red,-latex,out=110,in=180] (back8-kiswa.center) to (root-kiswa.south);
\node[rectangle,name=hook,at={(back8-kiswa)}]{};
\end{tikzpicture}
\end{adjustbox}
\end{subfigure}\hfill
\begin{subfigure}[t]{.675\linewidth}
\centering
\begin{adjustbox}{valign=t}
\begin{tikzpicture}[remember picture]
\Tree [.\node[branch,name=root-gnash]{}; [ [ [.\node[branch,name=back8-gnash]{B}; [ [ \node[leaf,name=back64-gnash]{\phantom{0}}; [ 7 \node[leaf,name=back128-gnash]{\phantom{0}}; ] ] [ 6 7 ] ] [ 6 7 ] ] [ 7 \node[leaf,name=back16-gnash]{\phantom{0}}; ] ] [ 6 7 ] ] [ 6 7 ] ]
\draw[red,-latex] (back128-gnash.center) to[out=110,in=180] (root-gnash.south);
\draw[red,-latex] (back64-gnash.center) to[out=110,in=180] (root-gnash.south);
\draw[red,-latex] (back16-gnash.center) to[out=110,in=180] (root-gnash.south);
\end{tikzpicture}
\end{adjustbox}
\end{subfigure}
\begin{subfigure}[t]{.325\linewidth}
\caption{$\fldr$ ($k=4$)}
\end{subfigure}\hfill
\begin{subfigure}[t]{.675\linewidth}
\caption{ALDR ($K=7$)}
\end{subfigure}
\caption{Isomorphic DDG trees for $\bp = (6/13, 7/13)$.}
\label{fig:ddg-isomorphism}
\end{figure}

Let $k \geq 4$, $m \defeq (2^{k}-3)$, and
\begin{equation}
\bp \defeq \left(\frac{2^{k-1}-1}{m}, \frac{2^{k-1}-2}{m}\right).
\end{equation}
Applying \cref{eq:ddg-tree-toll,eq:fldr-cost} gives
\begin{align*}
&\toll{\fldr[\bp]}
\\
&\defeq
  \begin{aligned}[t]
    &\frac{2^k}{m} \left(
        \nuu{ \frac{3}{2^k} }
        + \nuu{ \frac{2^{k-1}-1}{2^k} }
        + \nuu{ \frac{2^{k-1}-2}{2^k} }
      \right)
    \\
    &- H(\bp)
  \end{aligned}
\\
&=
\frac{2^k}{2^k-3} \left[
    \begin{aligned}
        &\left(\frac{k-1}{2^{k-1}} + \frac{k}{2^k}\right) \\
        &+ \left(\frac{3}{2} - \frac{k+2}{2^k}\right) \\\
        &+ \left(\frac{3}{2} - \frac{k+1}{2^{k-1}}\right)
      \end{aligned}
      \right]
    - \Hb{\frac{2^{k-1}-1}{2^k-3}}
\\
&=
\frac{2^k}{2^k-3} \left[ 3 \frac{2^k-2}{2^k} \right]
- \Hb{\frac{2^{k-1}-1}{2^k-3}}
\\
&> 3 - 1
= 2.
\end{align*}
Because $k \geq 4$, we have $2^{2k-2} = 2^{k-2}m + 3\cdot2^{k-2}$ and
$3\cdot2^{k-2} < 2^k-3 \eqdef m$, so $c_{2k-2}=2^{k-2}$ and $c_K = 2^{K-k}$ for
$K = k,\dots,2k-2$, which implies that $\bq_k = \cdots = \bq_{2k-2}$
(cf.~\cref{eq:c-case-even}), and so the corresponding $\aldr[\bp, K]$ tolls are
equal.
For $K = 2k-1$, we have $c_{2k-1} = 2^{k-1}+1$, and every $a_i$ has a bit-length
of at most $k-1$, so there are no carries in any product $a_i \times c_{2k-1}$.
The toll of $\aldr[\bp, 2k-1]$ is thus also equal to the $\fldr[\bp]$ toll, by
the no-carry case of \cref{theorem:aldr-leq-fldr}.
\end{proof}

\Cref{fig:ddg-isomorphism} illustrates the argument from the preceding proof,
where $\aldr[\bp, 2k-1]$ is isomorphic to $\fldr[\bp]$,
unrolling the back edge at depth $k-1$ once.
Although the toll of $\fldr[\bp]$ exceeds $2$ for each of these distributions,
for larger $k$ it approaches $2$:
\begin{align*}
\toll{\fldr[\bp]}
&= \toll{\aldr[\bp, 2k-1]} \\
&< 2 + (k+3)2^{1-k}
\end{align*}
in accordance with \cref{theorem:aldr-generic-bound}.

\section{Further Entropy Properties}
\label{sec:properties}

\Cref{fig:fldr478} suggests that $\aldr[\bp,K]$ can
``interpolate'' between $\fldr[\bp]$ (when $K = k$) and the
entropy-optimal $\ky[\bp]$ (for sufficiently large $K$).
This section studies the properties of $\aldr[\bp,K]$ as $K$ increases.
\begin{itemize}[itemsep=5pt]
\item \Cref{sec:properties-ky} characterizes distributions
$\bp$ for which $\aldr[\bp,K]$ is entropy optimal for
some depth $K$.
\item \Cref{sec:properties-fldr} establishes that the entropy cost
of $\aldr[\bp,K]$ is upper bounded by that of $\fldr[\bp] \equiv \aldr[\bp,k]$,
providing a necessary and sufficient condition for strict inequality.
\end{itemize}

\subsection{Comparison of ALDR and Entropy-Optimal Sampling}
\label{sec:properties-ky}

The first result shows that \cref{theorem:aldr-2k-toll-two} is tight and
that for certain target distributions $\bp$,
$\aldr[\bp,K]$ fails to be entropy optimal for any $K$.

\begin{proposition}
\label{proposition:aldr-large-gap}
For every $\epsilon > 0$, there exists a discrete probability distribution $\bp$
such that $\toll{\ky[\bp]} < \epsilon$
and $\toll{\aldr[\bp,K]} > 2 - \epsilon$ for all $K$.
\end{proposition}

\begin{proof}
It suffices to assume $\epsilon$ is a positive power of $1/2$,
by setting $\epsilon \gets \min\left(1/2, 2^{\floor{\log(\epsilon)}}\right)$ if needed.
Suppose $\bp = (a_1/m, a_2/m, \dots, a_n/m)$ is given by
\begin{align*}
n \defeq 2 + 2/\epsilon, \qquad
m \defeq 6 / \epsilon
\\
a_1 = a_2 = a_3 \defeq 1, \;
a_4 = \dots = a_n \defeq 3.
\end{align*}
The outcomes whose labels are $4, \dots, n$ have probability
$p_4 = \dots = p_n = \epsilon/2$ and total probability $1 - \epsilon/2$.
Then
\begin{equation*}
\toll{\ky[\bp]} = 3(\epsilon/6)\trel{\epsilon / 6} < \epsilon,
\end{equation*}
whereas
\begin{equation*}
\toll{\aldr[\bp,K]}
\geq (1-\epsilon/2) \trelfldr{A_n, M}
\geq 2 - \epsilon
\end{equation*}
by \cref{lemma:aldr-pow-two-toll}.
\end{proof}

For a distribution $\bp$ where almost all of the probabilities are powers
of two, the entropy toll of $\aldr[\bp,K]$ samplers can remain arbitrarily
close to two with increasing DDG tree depth $K$.

\begin{figure}[t]
\centering
\includegraphics[width=\linewidth]{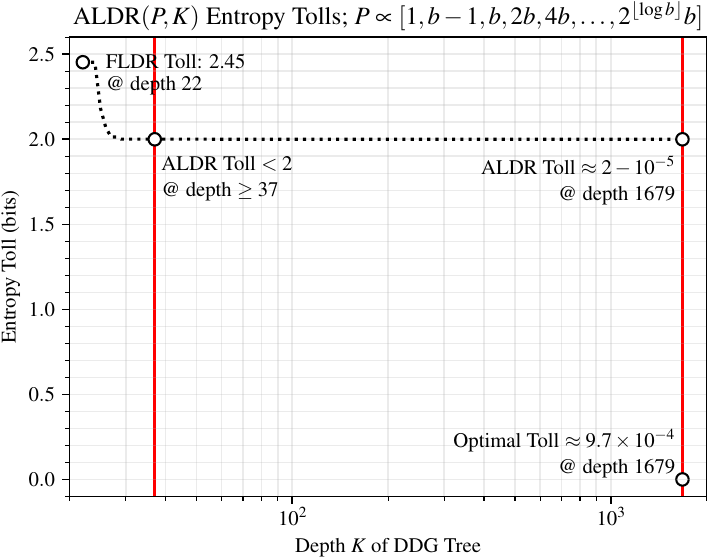}
\caption{Entropy toll of $\aldr[\bp,K]$ does not necessarily converge to
the optimal toll as $K$ increases.
Contrast to \cref{fig:fldr478}, where $\aldr[\bp,K]$ interpolates
between the $\fldr$ and entropy-optimal samplers as $K$ increases.
}
\label{fig:fldr-toll-1669}
\end{figure}

\begin{figure*}
\centering
\tikzset{level distance=20pt}
\tikzset{hidden/.style={color=gray, inner sep=1pt, minimum height=10pt}}
\captionsetup[subfigure]{skip=0pt,belowskip=10pt}
\begin{subfigure}[t]{.25\linewidth}
\caption{Unrolled Zero Times}
\label{fig:fldr-tree-unrolling-0}
\centering
\tikzset{sibling distance=8pt}
\begin{tikzpicture}
\Tree[ 4 [ [ R 1 ] R ] ]
\end{tikzpicture}
\end{subfigure}%
\begin{subfigure}[t]{.25\linewidth}
\caption{Unrolled One Time}
\label{fig:fldr-tree-unrolling-1}
\centering
\tikzset{sibling distance=2pt}
\begin{tikzpicture}
\Tree[ 4 [ [ [.\node[hidden]{R}; 4 [ [ R 1 ] \node(nleft){R}; ] ] 1 ] [.\node[hidden]{R}; 4 [ [ \node(nright){R}; 1 ] R ] ] ] ]
\node[draw,red,fit=(nleft)(nright),inner sep=1pt]{};
\end{tikzpicture}
\end{subfigure}%
\begin{subfigure}[t]{.5\linewidth}
\caption{Unrolled Two Times}
\label{fig:fldr-tree-unrolling-2}
\centering
\tikzset{sibling distance=2pt}
\begin{tikzpicture}
\Tree[ 4 [ [ [.\node[hidden]{R}; 4 [ [ [.\node[hidden]{R}; 4 [ [ R 1 ] R ] ] 1 ] [.\node[hidden]{R}; \node(nleft){4}; [ [ R 1 ] R ] ] ] ] 1 ] [.\node[hidden]{R}; 4 [ [ [.\node[hidden](ntop){R}; 4 [ [ R \node(nbottom){1}; ] [.\node(nright){R}; ] ] ] 1 ] [.\node[hidden]{R}; 4 [ [ R 1 ] R ] ] ] ] ] ]
\node[draw,red,fit=(nleft)(ntop)(nright)(nbottom),inner sep=1pt]{};
\end{tikzpicture}
\end{subfigure}
\caption{DDG tree of $\fldr$ sampler with output distribution $\bp = (1/5, 4/5)$ unrolled zero, one, and two times.
Leaf nodes labeled R (reject) indicate a back edge to the root of the tree.
In panels \subref{fig:fldr-tree-unrolling-1} and \subref{fig:fldr-tree-unrolling-2},
internal nodes labeled R (in gray) identify the locations of the unrolling.}
\label{fig:fldr-tree-unrolling}
\end{figure*}

\begin{example}[ALDR-KY gap]
\Cref{fig:fldr-toll-1669} shows the tolls of
$\fldr$, $\aldr$, and entropy-optimal DDG trees for
\begin{equation*}
\bp =  \left(\frac{1}{m}, \frac{b-1}{m}, \frac{b}{m}, \frac{2b}{m}, \frac{4b}{m}, \dots, \frac{2^{10}b}{m}\right),
\end{equation*}
where $b=1669$ and $m = 2^{11}b$.
Following the example in \cref{proposition:aldr-large-gap}---for which many
probabilities that are powers of two comprise most of the probability---this
choice ensures that $1-2^{-11}$ of the probability consists of
powers of two, which drives $\toll{\ky[\bp]}$ to near zero
(more precisely, $9.7 \times 10^{-4}$ bits),
whereas the toll of $\aldr[\bp,K]$ remains near $2$ for all $K$.
Additionally, $b=1669$ is a prime number such that the order of 2 modulo $b$ is
$b-1$, which ensures that the entropy-optimal sampler $\ky[\bp]$ has a large
depth of $b-1 + \ceil{\log(b)} = 1679$.
The $\fldr[\bp]$ sampler has depth $k=22$ and toll $2.45$ bits.
The $\aldr[\bp,K=37]$ sampler has a toll that is strictly
less than two bits (cf.~\cref{theorem:aldr-2k-toll-two}), and it remains
slightly below this value for all depths $K = 37, \dots, 1679$ without
ever reaching the optimal toll of nearly zero.
\end{example}

To address the question of when $\aldr[\bp,K]$ can be entropy optimal, we recall
that \cref{theorem:knuth-yao} fully characterizes entropy optimality in terms of
the distribution of leaf nodes across levels in a DDG tree.
To apply \cref{theorem:knuth-yao} to $\aldr[\bp,K]$, it is necessary to
characterize the distribution of leaf nodes in $\aldr[\bp,K]$, which are
obtained by infinitely unrolling the back edges that correspond to the
rejection label in $\ky[\bq_K]$.
The following theorem provides a method for directly counting the leaf nodes in
$\aldr[\bp,K]$ in terms of the leaf nodes in $\ky[\bq_K]$.
The proof is in~\cref{appx:proof-theorem-generating-fn}
of the online supplement.

\begin{theorem}[Counting leaves]
\label{theorem:generating-fn}
The number of leaf nodes with label $i$ at depth $d$ in the unrolled tree
$\fldr[A_1,\dots,A_n]$ can be computed in terms of the leaves in the unexpanded
tree $\ky[\bq]$ as the coefficients of a generating function (i.e., formal power
series):
\begin{align*}
&\sum_{d=0}^{\infty} \leaves{\fldr*[A_1,\dots,A_n]}{d}{i} z^d
= \frac{\sum_{d=0}^{\infty} \leaves{\ky*[\bq]}{d}{i} z^d}
  {1 - \sum_{d=0}^{\infty} \leaves{\ky*[\bq]}{d}{0} z^d}
\\
&= \frac{\sum_{d=0}^K \epsd{A_i/2^K}z^d}{1 - \sum_{d=0}^K \epsd{A_0/2^K}z^d}
\\
&= \sum_{j=0}^{\infty}
  \left\lbrace
    \left(\sum_{d=0}^K\epsd{\frac{A_i}{2^K}}z^d\right)
    \left(\sum_{d=0}^K \epsd{\frac{A_0}{2^K}}z^d\right)^j
  \right\rbrace,
\end{align*}
where $(1-f(z))^{-1} \defeq \sum_{j=0}^{\infty} f(z)^j$ for
a generating function $f(z)$.
\end{theorem}

\begin{corollary}[Optimal FLDR rejection]
\label{corollary:fldr-power-of-two}
If the rejection probability $q_0 \defeq A_0/2^K$ in $\fldr$
is neither zero nor a power
of two, then the FLDR tree is not entropy optimal, i.e.,
$\toll{\fldr[A_1,\dots,A_n]} > \toll{\ky[\bp]}$.
\end{corollary}

\begin{proof}
If $A_0/2^K$ is neither zero nor a power of two, then $A_0$ must have
at least two set bits, say
$\epsd[d_1]{A_0/2^K} = \epsd[d_2]{A_0/2^K} = 1$.
Also, let $d_3$ be any set bit in $A_1 / 2^K$, i.e.,
$\epsd[d_3]{A_1/2^K} = 1$.
Using $\left[z^D\right]\sum_d g_d z^d \defeq g_D$ to denote the $z^D$-coefficient of
a given generating function, the number of leaf nodes with label $1$ at depth
$d_1 + d_2 + d_3$ (cf.~\cref{fig:fldr-tree-unrolling-2}, depth $8$) in the fully
unrolled tree is
\begin{align*}
&\left[z^{d_1+d_2+d_3}\right]
  \frac{\sum_{d=0}^K z^d \epsd{A_1/2^K}}{1 - \sum_{d=0}^K z^d \epsd{A_0/2^K}}
\\
&\geq \left[z^{d_1+d_2+d_3}\right] \frac{z^{d_3}}{1 - z^{d_1} - z^{d_2}} \\
&\geq \left[z^{d_1+d_2+d_3}\right] z^{d_3}(z^{d_1} + z^{d_2})^2 \\
&\geq \left[z^{d_1+d_2+d_3}\right] 2z^{d_1+d_2+d_3} \\
&= 2,
\end{align*}
which is entropy suboptimal by \cref{theorem:knuth-yao}.
\end{proof}

\begin{example}[Suboptimal FLDR rejection]
\Cref{fig:fldr-tree-unrolling} shows DDG trees obtained by unrolling
the back edges of $\fldr[\bp]$ zero, one, and two times,
where $\bp \defeq (1/5, 4/5)$.
Nodes labeled $1$ and $4$
represent the outcomes with $1/5$ and $4/5$ probability, respectively,
as in \cref{fig:ddg-ky-fldr}.
Nodes labeled R represent ``reject nodes'', which are back edges to
the root of the tree.
The twice-unrolled tree in \cref{fig:fldr-tree-unrolling-2} has two leaves
with label $4$ at depth $6$, and two leaves with label $1$ at depth $8$,
corresponding to two possible orders of traversing the $1/4$-- and
$1/8$--probability reject nodes in \cref{fig:fldr-tree-unrolling-0}.
This type of duplication occurs for any FLDR (or ALDR) tree whose rejection
probability is not a power of two.
The proof of \cref{corollary:fldr-power-of-two} uses this same idea of
traversing the same two rejection nodes in a different order to identify a
depth with two identical labels.
\end{example}

We now characterize which probability distributions can have
entropy-optimal ALDR trees.

\begin{theorem}[Characterization of entropy-optimal ALDR trees]
\label{theorem:aldr-ky-match-bounded}
Consider coprime integer weights $(a_1,\dots,a_n)$
whose sum $m$ is not a power of two.
Write $m = 2^u x$ for odd $x>1$ and $u \geq 0$, and
write $\ell$ for the order of $2 \bmod x$,
so that a standard construction of the entropy-optimal \citeauthor{knuth1976}
tree has depth $u + \ell$ (cf.~\cref{fig:fldr-ky-match-even-optimal}).
The following statements are equivalent.
\begin{enumerate}[label=(\labelcref{theorem:aldr-ky-match-bounded}.\arabic*)]
\item \label{item:aldr-ky-match-bounded-k}
$\aldr[\bp,K]$ is entropy optimal
for some depth $K \leq u+\ell$.

\item \label{item:aldr-ky-match-bounded-k-ul}
$\aldr[\bp,u+\ell]$ is entropy optimal.

\item \label{item:aldr-ky-match-bounded-gen-eq}
For each $i=1,\dots,n$,
the binary expansions of the probabilities $p_i$ and $q_i$ satisfy
\begin{equation*}
(1 - z^\ell) \sum_d \epsd{p_i} z^d = \sum_d \epsd{q_i} z^d.
\end{equation*}

\item \label{item:aldr-ky-match-bounded-gen}
For each $i=1,\dots,n$,
all the coefficients of the generating function
$(1 - z^\ell) \sum_d \epsd{p_i} z^d$ are nonnegative.

\item \label{item:aldr-ky-match-bounded-p}
For each $i=1,\dots,n$,
the binary expansion of $p_i$ satisfies
$\epsd{p_i} \leq \epsd[d+\ell]{p_i}$
for all $d \ge 1$.
\qedhere
\end{enumerate}
\end{theorem}

\begin{proof}
By \cref{corollary:fldr-power-of-two}, the rejection probability of
the ALDR tree must be a power of two in order to achieve
entropy optimality.
The smallest depth at which the rejection probability $1-cm/2^K$ becomes a power
of two is $K = u + \ell$, where $M = 2^{u+\ell} - 2^u$, so
\labelcref{item:aldr-ky-match-bounded-k,item:aldr-ky-match-bounded-k-ul} are
equivalent.

\Cref{theorem:generating-fn,theorem:knuth-yao} show that
\labelcref{item:aldr-ky-match-bounded-k-ul,item:aldr-ky-match-bounded-gen-eq} are
equivalent, because the distribution of labeled leaves at different depths in
the tree uniquely determines the set of entropy-optimal trees.

For the equivalence of
\labelcref{item:aldr-ky-match-bounded-gen-eq,item:aldr-ky-match-bounded-gen}, the
forward implication follows from the nonnegativity of $\epsd{q_i}$.
For the reverse implication, consider the following properties of
$(1 - z^\ell) \sum_d \epsd{p_i} z^d$:
its only nonzero coefficients appear at $0 < d \leq u+\ell$,
every coefficient is at most $1$,
and when evaluated at $z \mapsto 1/2$, it becomes $q_i \defeq A_i / 2^K$.
Given the additional assumption of nonnegative coefficients, uniqueness of the
finite binary representation of dyadic numbers then implies that
$(1 - z^\ell) \sum_d \epsd{p_i} z^d = \sum_d \epsd{A_i / 2^K} z^d$,
which establishes the reverse implication.

Lastly, the equivalence of
\labelcref{item:aldr-ky-match-bounded-gen,item:aldr-ky-match-bounded-p}
follows from
$(1 - z^\ell) \sum_d \epsd{p_i} z^d
= \sum_d (\epsd{p_i} - \epsd[d-\ell]{p_i})$.
\end{proof}

\begin{remark}[Full characterization of entropy-optimal ALDR trees]
\label{remark:aldr-ky-match-general}
In the notation of \cref{theorem:aldr-ky-match-bounded},
$\aldr[\bp,K]$ can match the entropy-optimal tree
$\ky[\bp]$ within
depth $K \leq u+\Lambda\ell$ if and only if
there is some $\lambda \in \{1, \dots, \Lambda\}$ such that
for all $d \geq 1$ and $1 \leq i \leq n$,
the bits in the binary expansion of $p_i$ satisfy
$\epsd{p_i} \leq \epsd[d+\lambda\ell]{p_i}$.
This result follows directly from the proof of
\cref{theorem:aldr-ky-match-bounded},
but allowing any rejection probability $2^{-\lambda\ell}$
and replacing $1/(1-z^\ell)$ with $1/(1-z^{\lambda\ell})$.
If the inequality holds for some $\Lambda$,
then it also holds for $\lambda = \ceil{u / \ell}$.
This property is similar to the requirement that no probability be dyadic,
but stronger.
For example, it is impossible to obtain an entropy-optimal representation
of $5/6 = (0.1\overline{10})_2$ using any finite DDG tree with a
back edge to the root.
\end{remark}

\begin{corollary}[Entropy-optimal ALDR trees for odd denominators]
\label{proposition:fldr-ky-match-odd}
In the notation of \cref{theorem:aldr-ky-match-bounded},
if the distribution $\bp$ has odd denominator $m$,
then $\aldr[\bp,\ell]$ is the entropy-optimal tree.
\end{corollary}

\begin{proof}
Apply \cref{theorem:aldr-ky-match-bounded}, noting that $u=0$ and
the binary expansion of each $p_i$ repeats with period $\ell$.
\end{proof}

\begin{example}[ALDR matching \citeauthor{knuth1976}]
\begin{figure}[b]
\centering
\tikzset{sibling distance=1pt, level distance=10pt}
\begin{subfigure}[t]{.31\linewidth}
\centering
\begin{adjustbox}{valign=t}
\begin{tikzpicture}
\Tree[.\node[branch,name=root]{}; [ [ \node[leaf,name=back1]{\phantom{0}}; [ 1 3 ] ] [ 2 3 ] ] [ 4 \node[leaf,name=back2]{\phantom{0}}; ] ]
\draw[red,-latex,out=90,in=0] (back2.center) to (root);
\draw[red,-latex,out=90,in=180] (back1.center) to (root);
\end{tikzpicture}
\end{adjustbox}
\end{subfigure}%
\begin{subfigure}[t]{.35\linewidth}
\centering
\begin{adjustbox}{valign=t}
\begin{tikzpicture}
\Tree[.\node[branch,name=root]{}; [ [ [ \node[leaf,name=back]{\phantom{0}}; 2 ] [ [ 1 3 ] 1 ] ] [ 2 4 ] ] [ 3 4 ] ]
\draw[red,-latex,out=90,in=180] (back.center) to (root);
\end{tikzpicture}
\end{adjustbox}
\end{subfigure}%
\begin{subfigure}[t]{.33\linewidth}
\centering
\begin{adjustbox}{valign=t}
\begin{tikzpicture}
\Tree [ [.\node[name=root1,branch,fill=blue]{}; [ [ [ \node[leaf,name=back1]{\phantom{0}}; \node[leaf,name=back2]{\phantom{0}}; ] [ 1 3 ] ] [ 1 2 ] ] [ 2 4 ] ] [.\node[name=root2,branch]{}; 3 4 ] ]
\draw[red,-latex,out=90,in=180] (back1.center) to (root1);
\draw[red,-latex,out=45,in=180] (back2.center) to (root2);
\end{tikzpicture}
\end{adjustbox}
\end{subfigure}

\begin{subfigure}{.33\linewidth}
\caption{$\fldr[\bp]$}
\label{fig:fldr-ky-match-even-fldr}
\end{subfigure}%
\begin{subfigure}{.33\linewidth}
\caption{$\aldr[\bp,5]$}
\label{fig:fldr-ky-match-even-amplified}
\end{subfigure}%
\begin{subfigure}{.33\linewidth}
\caption{$\ky[\bp]$}
\label{fig:fldr-ky-match-even-optimal}
\end{subfigure}

\caption{Comparison of $\fldr$ (depth 4), ALDR (depth 5),
  and \citeauthor{knuth1976} DDG trees for $\bp = (1, 2, 3, 4)/10$.
  The DDG tree representations in \subref{fig:fldr-ky-match-even-amplified} and \subref{fig:fldr-ky-match-even-optimal}
  are isomorphic and identical upon infinite unrolling.
  }
\label{fig:fldr-ky-match-even}
\end{figure}

Suppose $\bp = (1, 2, 3, 4)/10$.
In the notation of \cref{theorem:aldr-ky-match-bounded}, $m=10$,
$u=1$, $x=5$, and the order of $2 \bmod 5$ is $\ell=4$.
The binary expansion of each probability repeats with period $\ell$ and
with zero preperiod (i.e., the periodic repetition starts immediately with
the first set bit).
\Cref{fig:fldr-ky-match-even} shows that
$\aldr[\bp, 5]$ matches the entropy-optimal sampler $\ky[\bp]$;
its rejection node points to the root, whereas a typical representation of
the entropy-optimal uses two back edges that target the
children of the root.
\end{example}

\subsection{Comparison of ALDR and FLDR}
\label{sec:properties-fldr}

This section establishes conditions under which $\aldr[\bp,K]$ has a lower
entropy cost than $\fldr[\bp] \equiv \aldr[\bp,k]$
for a given $K > k$.
The first result shows that the toll does not necessarily decrease with $K$.

\begin{proposition}[ALDR toll not monotonic in depth]
\label{proposition:aldr-not-monotonic-decrease}
There exists a discrete probability distribution $\bp$ and depth $K$ such that
$\toll{\aldr[\bp, K]} < \toll{\aldr[\bp, K+1]}$.
\end{proposition}

\begin{proof}
Set $\bp = (4/19, 7/19, 8/19)$.
Then
\begin{align*}
\expect{\cost{\aldr[\bp, 10]}} &= \frac{3038}{1007} = 3.01608\dots \\
\expect{\cost{\aldr[\bp, 11]}} &= \frac{6150}{2033} = 3.02508\dots
\end{align*}
The tolls just subtract $H(\bp)$ from the entropy cost,
so they have the same relationship, as shown in \cref{fig:fldr478-tolls}.
This non-monotonicity is dictated by the irregularity of
$\trel{A_0/2^K}$ as $K$ varies.
The rejection outcome in $\ky[\bq_{10}]$ has
probability $17/1024$, which has a small relative toll
$\trel{17/1024} \approx 0.32$, whereas the rejection outcome in
$\ky[\bq_{11}]$ has probability $15/2048$,
which has a much larger relative toll $\trel{15/2048} \approx 1.64$.
\end{proof}

To compare the entropy costs of $\fldr[\bp]$ and
$\aldr[\bp,K]$ in general, we must relate
the terms $\nuu{a_i/2^k}$ and $\nuu{c_K a_i / 2^K}$, which, respectively,
make up the costs of the two trees.
This comparison requires an analysis of the behavior of
$\nu$ under multiplication.
To better describe the properties of $\nu$, we generalize
its domain from the set of positive real numbers to a more natural domain.

\begin{definition}
\label{definition:nu-entropy}
The \textit{``new'' entropy function}
$\nu(x) \defeq \sum_{d=0}^{\infty} d \epsd{x}2^{-d}$
from \citet{knuth1976}
can be extended to any $\nat$-ordered and $\real$-valued sequence,
which in the general case is represented by the
coefficients of a formal Laurent series
$g = \sum_{d = D}^{\infty} g_d z^d \in \real((z))$
for some $D \in \integers$.
On this domain, we define $\nu(g) \defeq \sum_d d g_d z^d$.
These formal Laurent series can be converted to real numbers by evaluating with
$z \mapsto 1/2$, which is consistent with the definition of $\nu$ on real
numbers, because the diagram
\begin{equation}
\begin{tikzcd}
  {\real_{\geq 0}} &&& {\real((z))} \\
  \real &&& {\real((z))}
  \arrow["{x \mapsto \sum_d \epsd{x} z^d}", tail, from=1-1, to=1-4]
  \arrow["\nu"', from=1-1, to=2-1]
  \arrow["\nu", from=1-4, to=2-4]
  \arrow["{z \mapsto 1/2}", rightharpoonup, from=2-4, to=2-1]
\end{tikzcd}
\end{equation}
commutes.
The series evaluation $z \mapsto 1/2$ is a partial function because it may
fail to converge for some $g$.
However, convergence occurs if $g_d/2^d = O(d^{-2})$, which holds for
every sequence in this paper.
\end{definition}

\begin{remark}
\label{remark:toll-per-weight}
The generalized $\nu$ entropy can be used to directly compute the toll of a tree
given the generating function for its leaf counts.
For example, applying $\nu$ to a generating function from
\cref{theorem:generating-fn} yields
\begin{equation*}
\begin{aligned}
&\left(\frac{A_i}{M}\right)\trelfldr{A_i, M}
\\
&\quad=\begin{aligned}[t]
  &\left[
    \nuu{\frac{\sum_{d=0}^K \epsd{A_i/2^K}z^d}{1 - \sum_{d=0}^K \epsd{A_0/2^K}z^d}}
      \right]_{z \mapsto 1/2}
    \\
  &- \Ho{A_i / M}
  \end{aligned}
\end{aligned}
\end{equation*}
as an alternative expression for the toll contributions in a FLDR tree
(cf.~\cref{definition:fldr-toll}).
\end{remark}

\begin{definition}[Carries]
\label{definition:carries}
For nonnegative $x,y \in \real_{\geq 0}$ and $\star \in \set{+, \times}$, the operation
$x \star y$ is said to \textit{carry} iff
$\sum_d \epsd{x \star y} z^d \neq
\left( \sum_d \epsd{x} z^d \right) \star \left( \sum_d \epsd{y} z^d \right)$.
\end{definition}

\begin{remark}
\label{remark:carries}
Given $x \defeq (0.x_1x_2x_3\ldots)_2$ and $y \defeq (0.y_1y_2y_3\ldots)_2 \in [0,1)$,
the addition $z = x + y = (0.z_1z_2z_3\ldots)_2$ does not carry if
and only if $z_i = x_i + y_i$ for $i \geq 1$ (cf.~\citep[Eqs.~(2.23--2.24)]{knuth1976}).
In other words, the addition carries if and only if one of the
two conditions holds: \begin{inlinelist}[label=(\alph*)]
\item $z$ is dyadic when $x$ is not; or
\item there exists $i \geq 1$ such that $x_i = y_i = 1$.
\end{inlinelist}
The multiplication $z = xy = (0.z_1z_2z_3\ldots)_2$ does not carry
if and only if $z_1 = 0$ and $z_j = \sum_{i=1}^{j-1}x_i y_{j-i}$.
(All binary expansions in these definitions are concise.)
\end{remark}

We now characterize the properties of $\nu$ needed for \cref{theorem:aldr-leq-fldr},
building on \citet[Eq.~(2.23)]{knuth1976} to additionally handle multiplication,
and generalizing the algebraic notion of a ring derivation.

\begin{lemma}[{$\nu$} is a subadditive subderivation]
\label{lemma:nu-derivation}
For real $x, y \geq 0$,
the $\nu$ entropy satisfies
\begin{align}
\nu(x + y) &\leq \nu(x) + \nu(y),
\label{eq:nu-subadditive}\\
\nu(x y) &\leq x \nu(y) + y \nu(x),
\label{eq:nu-derivation}
\end{align}
with equality in \cref{eq:nu-subadditive} iff $x + y$
does not carry,
and equality in \cref{eq:nu-derivation} iff $x \times y$
does not carry.
These inequalities have the same form as additivity and Leibniz's law (i.e., the
product rule), except that `$=$' is replaced by `$<$' whenever the corresponding
operation carries.
\end{lemma}

\begin{proof}
Consider the extension of the $\nu$ entropy to formal Laurent series
$\nuu{\sum_d g_d z^d} = \sum_d d g_d z^d$
from \cref{definition:nu-entropy}.
This $\nu$ is a derivation on formal Laurent series,
\begin{align*}
  \nu(f(z) + g(z)) &= \nu(f(z)) + \nu(g(z)), \\
  \nu(f(z)g(z)) &= f(z) \nu(g(z)) + g(z) \nu(f(z)),
\end{align*}
because $\nu = z \frac{\diff}{\diff z}$.
%
In using bitstring conversions from real numbers, the only difference is
carries, which gives the following two inequalities for sums
and products, respectively:
\begin{align*}
\nu(x + y)
&= \left[ \nuu{\sum_d \epsd{x + y} z^d} \right]_{z \mapsto 1/2} \\
&\leq \left[ \nuu{\sum_d \epsd{x} z^d+ \sum_d \epsd{y} z^d } \right]_{z \mapsto 1/2} \\
&= \begin{aligned}
    &\left[ \nuu{\sum_d \epsd{x} z^d} \right]_{z \mapsto 1/2}
    \\
    &+ \left[ \nuu{\sum_d \epsd{y} z^d} \right]_{z \mapsto 1/2}
  \end{aligned}
\\
&= \nu(x) + \nu(y),
\shortintertext{and}
\nu(x y)
&= \left[ \nuu{\sum_d \epsd{x y} z^d} \right]_{z \mapsto 1/2} \\
&\leq \left[ \nuu{\sum_d \epsd{x} z^d \sum_e \epsd[e]{y} z^e} \right]_{z \mapsto 1/2} \\
&= \left[ \begin{aligned}
  &\sum_d \epsd{x} z^d \nuu{\sum_e \epsd[e]{y} z^e}
  \\
  &+ \sum_e \epsd[e]{y} z^e \nuu{\sum_d \epsd{x} z^d}
  \end{aligned}\right]_{z \mapsto 1/2} \\
&= \left[ \begin{aligned}
    &\sum_d \epsd{x} \nu(y) z^d
    \\
    &+ \sum_e \epsd[e]{y} \nu(x) z^e
    \end{aligned} \right]_{z \mapsto 1/2} \\
&= x \nu(y) + y \nu(x).
\end{align*}
These inequalities are strict if and only if $x + y$ or $x \times y$,
respectively, carries.
\end{proof}

The main theorem of this section shows that the non-monotonicity
identified in \cref{proposition:aldr-not-monotonic-decrease} never occurs when
comparing the entropy cost of $\fldr[\bp]$ with that of
$\aldr[\bp, K]$ for $K > k$.
In particular, since $\fldr[\bp] \equiv \aldr[\bp, k]$ is the first sampler
in the $\aldr[\bp,K]$ family ($K=k,k+1,\dots$) of rejection samplers for
the target $\bp$, the entropy costs of all members are upper bounded by
that of the first member.

\begin{theorem}[Comparison of $\aldr$ and $\fldr$ entropy cost]
\label{theorem:aldr-leq-fldr}
For every distribution $\bp$, the entropy cost
of $\aldr$ is upper bounded by that of $\fldr$:
\begin{align*}
\expect{\cost{\aldr[\bp,K]}} \leq \expect{\cost{\fldr[\bp]}}
&&(K \geq k),
\end{align*}
with equality if and only if no
multiplication $c_K \times a_i$ carries for each $i=0,\dots,n$.
\end{theorem}

\begin{listing*}[t]
\captionsetup{hypcap=false}
\begin{minipage}[t]{.495\linewidth}
\begin{algorithm}[H]
\caption{ALDR Preprocessing}
\label{alg:aldr-preprocess}
\begin{algorithmic}[1]
\Require{Coprime positive integers $(a_1,\dots,a_n)$}
\Ensure{Data structures used for sampling (\cref{alg:aldr-sample})}
\State \textbf{initialize} $A$ $\mathbf{int}[n+1]$    \Comment{amplified weights array}
\State $m \gets a_1+\dots+a_n$                        \Comment{sum of weights}
\State $K \gets 2\ceil{\log(m)}$                      \Comment{default amplification depth}
\State $c \gets \floor{2^K/m}$                        \Comment{amplification factor}
\State $A[0] \gets 2^K - c \cdot m$                   \Comment{amplified reject weight}
\State $A[i] \gets c \cdot a_i$ ($i=1,\dots,n$)       \label{algline:aldr-preprocess-prod}
                                                      \Comment{amplified weights}
\State \textbf{initialize} $L$ $\mathbf{int}[K+1]$    \Comment{leaves per depth array}
\State \textbf{initialize} $F$ $\mathbf{int}[]$       \Comment{flattened leaf labels list}
\For{$j \gets 0 ~\mathbf{to}~ K$}                     \Comment{for each level}\label{algline:aldr-preprocess-for-K}
  \For{$i \gets 0 ~\mathbf{to}~ n$}                   \Comment{for each label}\label{algline:aldr-preprocess-for-n}
    \If{$\floor{A[i]/2^{K-j}} \bmod 2 = 1$}           \Comment{leaf node}
      \State $L[j] \gets L[j] + 1$                    \Comment{update counter}
      \State $F.\hspace*{1pt}\mathit{append}(i)$      \Comment{store label}
    \EndIf
  \EndFor
\EndFor
\State \Return $L, F$
\end{algorithmic}
\end{algorithm}
\end{minipage}\hfill
\begin{minipage}[t]{.495\linewidth}
\begin{algorithm}[H]
\caption{ALDR Sampling}
\label{alg:aldr-sample}
\begin{algorithmic}[1]
\Require{%
  Data structures $L,F$ from \cref{alg:aldr-preprocess};\\
  Entropy source $\flip()$ providing i.i.d.~fair bits}
\Ensure{Sample from distribution given to \cref{alg:aldr-preprocess}}
\State $\set{d \gets 0;\ \ell \gets 0;\ v \gets 0}$       \Comment{depth, location, value}
\While{\textbf{true}}
  \If{$v < L[d]$}                                         \Comment{hit leaf node}
    \If{$F[\ell+v] = 0$}                                  \Comment{reject label}
      \State $\set{d \gets 0;\ \ell \gets 0;\ v \gets 0}$ \Comment{restart}
    \Else                                                 \Comment{accept label}
        \State \Return $F[\ell+v]$                        \Comment{return the label}
    \EndIf
  \EndIf
  \State $v \gets 2 \cdot (v - L[d]) + \flip()$           \Comment{visit random child}
  \State $\ell \gets \ell + L[d]$                         \Comment{increment location}
  \State $d \gets d + 1$                                  \Comment{increment depth}
\EndWhile
\end{algorithmic}
\end{algorithm}
\end{minipage}

\captionsetup{type=table,belowskip=10pt}
\centering
\caption{Maximum value of the sum $m$
of integers $(a_1,\dots,a_n)$ in the target distribution $\bp$
to guarantee that
\cref{alg:aldr-preprocess} does not overflow, under various settings of the
word size $w$ of the underlying word RAM computer.}
\label{table:arithmetic}
\footnotesize
\begin{tabular*}{\linewidth}{r@{\extracolsep{\fill}}ccc}
\hline
~                       & \multicolumn{3}{c}{$m$}\\ \cline{2-4}
\multicolumn{1}{c}{$w$} & \multicolumn{1}{c}{Single Word Arithmetic} & \multicolumn{1}{c}{Double Word Arithmetic} & \multicolumn{1}{c}{Quadruple Word Arithmetic} \rule{0pt}{2.5ex} \\  
8 bits                  & 16                                         & 256                                        & 65,025 \\
16 bits                 & 256                                        & 65,536                                     & 4,294,836,225 \\
32 bits                 & 65,536                                     & 4,294,967,296                              & 18,446,744,065,119,617,025 \\
64 bits                 & 4,294,967,296                              & 18,446,744,073,709,551,616                 & 340,282,366,920,938,463,426,481,119,284,349,108,225 \\\hline\hline
\end{tabular*}

\end{listing*}

\begin{proof}
Recall that $c_K = \floor{2^K/m}$.
We now apply \cref{lemma:nu-derivation} to bound parts of the entropy cost
of the ALDR tree in terms of corresponding parts of the FLDR tree.
For $1 \leq i \leq n$, \cref{eq:nu-derivation} directly yields
\begin{align}
\nuu{\frac{A_i}{2^K}}
&= \nuu{\frac{c_K}{2^{K-k}} \frac{a_i}{2^k}} \notag \\
&\leq \frac{c_K}{2^{K-k}} \nuu{\frac{a_i}{2^k}} + \frac{a_i}{2^k} \nuu{\frac{c_K}{2^{K-k}}}.
\label{eq:weight-nu-entropy}
\end{align}
For $i = 0$, the sum $A_0 + (2^k c_K - 2^K)$ does not carry because
$A_0 < 2^k$.
Then the no-carry case of \cref{eq:nu-subadditive} together with either case of
\cref{eq:nu-derivation} gives
\begin{align*}
&\nuu{\frac{A_0}{2^K}}+\nuu{\frac{2^k c_K - 2^K}{2^K}}\\
&= \nuu{\frac{A_0 + 2^k c_K - 2^K}{2^K}}
= \nuu{\frac{c_K a_0}{2^K}} \\
&= \nuu{\frac{c_K}{2^{K-k}} \frac{a_0}{2^k}}
\leq \frac{c_K}{2^{K-k}} \nuu{\frac{a_0}{2^k}} + \frac{a_0}{2^k} \nuu{\frac{c_K}{2^{K-k}}}.
\end{align*}
Applying additivity with $\nu(1) = 0$ shows that
\begin{equation*}
\nuu{\left(2^k c_K - 2^K\right)/2^K} = \nuu{(2^k c_K)/2^K},
\end{equation*}
so
\begin{align}
\nuu{\frac{A_0}{2^K}}
  \leq \frac{c_K}{2^{K-k}} \nuu{\frac{a_0}{2^k}}
    &+ \frac{a_0}{2^k} \nuu{\frac{c_K}{2^{K-k}}} \notag
    \\
    &- \nuu{\frac{c_K}{2^{K-k}}}.
\label{eq:rejection-nu-entropy}
\end{align}
Combining \cref{eq:weight-nu-entropy,eq:rejection-nu-entropy}
gives the desired bound on the total cost of ALDR:
\begin{align*}
&\expect{\cost{\aldr[\bp,K]}} \\
&= \frac{2^K}{M} \sum_{i=0}^{n} \nuu{\frac{A_i}{2^K}} \\
&\leq \begin{aligned}[t]
  &\frac{2^K}{M} \sum_{i=0}^{n}
    \left[ \frac{c_K}{2^{K-k}} \nuu{\frac{a_i}{2^k}}
    + \frac{a_i}{2^k} \nuu{ \frac{c_K}{2^{K-k}}} \right]
  \\
  &- \frac{2^K}{M}\nuu{\frac{c_K}{2^{K-k}}}
  \end{aligned}\\
&= \frac{2^K}{M} \sum_{i=0}^{n} \frac{c_K}{2^{K-k}} \nuu{\frac{a_i}{2^k}} \\
&= \frac{2^k}{m} \sum_{i=0}^{n} \nuu{\frac{a_i}{2^k}} \\
&= \expect{\cost{\fldr[\bp]}}.
\end{align*}

To establish the necessary and sufficient conditions for equality, it is sufficient
to note that every inequality in
\cref{eq:weight-nu-entropy,eq:rejection-nu-entropy} is an equality iff
the corresponding multiplication $c_K \times a_i$ does not carry, by
\cref{lemma:nu-derivation}.
Therefore, the final result is an equality iff
none of the multiplications $c_K \times a_i$ carry $(i=0,1,\dots,n)$.
\end{proof}

\section{Implementation}
\label{sec:implementation}

\Cref{alg:aldr-preprocess,alg:aldr-sample} show fast integer-arithmetic
implementations of the preprocessing and sampling steps from the high-level
description of ALDR in \cref{alg:fldr,alg:aldr},
using $K=2k$ for concreteness instead of an arbitrary rule $K=r(k)$.
The preprocessing method uses a flattened leaf list $F$, which requires
roughly half as much memory as an explicit DDG tree representation.
The cost of the $n$ multiplications on \cref{algline:aldr-preprocess-prod}
is $n\log(m)\log(\log m)$ operations using the
\citeauthor{harvey2021} method~\citep{harvey2021} or
$n (\log m)^{\log3}$ operations
using the \citeauthor{karatsuba1962} method~\citep{karatsuba1962}.
The nested loops in
\crefrange{algline:aldr-preprocess-for-K}{algline:aldr-preprocess-for-n}
require $2 n \log(m)$ iterations, where the logarithmic factor is incurred by
the need to iterate through each bit of the amplified weights.

\subsection{Numerics of Integer Arithmetic}
\label{sec:implementation-numerics}

To characterize the regime in which \cref{alg:aldr-preprocess} is guaranteed
to never overflow, consider a standard word RAM model with a word size of $w > 1$ bits,
so that each $a_i \leq 2^{w}-1$, for $i=1,\dots,n$.
The number of outcomes $n \leq 2^w-1$ is upper bounded by
the number of addressable memory cells,
although any typical application in high-precision settings
will have $n \ll 2^w$ (e.g., a 64-bit machine with 64GiB of available
memory can hold at most $2^{33} \ll 2^{64}$ outcomes).
It follows that $m \leq (2^w-1)^2$, $k \leq 2w$, and, with the
doubling amplification rule, $K = 2k \leq 4w$.
Because $m \leq 2^k$, all intermediate values in
\cref{alg:aldr-preprocess} are less than $2^{2k}$.
If $m \leq 2^{\floor{w/2}}$ requires half a machine word, then
\cref{alg:aldr-preprocess} operates using efficient single-word arithmetic.
If $m < 2^w$ requires one machine word, then \cref{alg:aldr-preprocess}
requires double-word arithmetic to avoid overflow.
In the worst case, $m = (2^w-1)^2 > 2^w$ requires two machine words,
which means \cref{alg:aldr-preprocess} requires at most quadruple-word
arithmetic to avoid overflow.
\Cref{table:arithmetic} show various values of $m$ that can be supported
without overflow
under various machine word sizes $w$ with single-, double-, and quadruple-word
integer arithmetic.
If very large values of $m$ are needed, then double-word arithmetic on a 64-bit
machine can be supported using, for example, the
\texttt{unsigned \_\_int128} type from GCC C compiler or
the \texttt{u128} type from Rust.
Arithmetic with these types is highly efficient on machines
with 64-bit architectures.

\begin{figure*}[t]
\centering
\includegraphics[width=\linewidth]{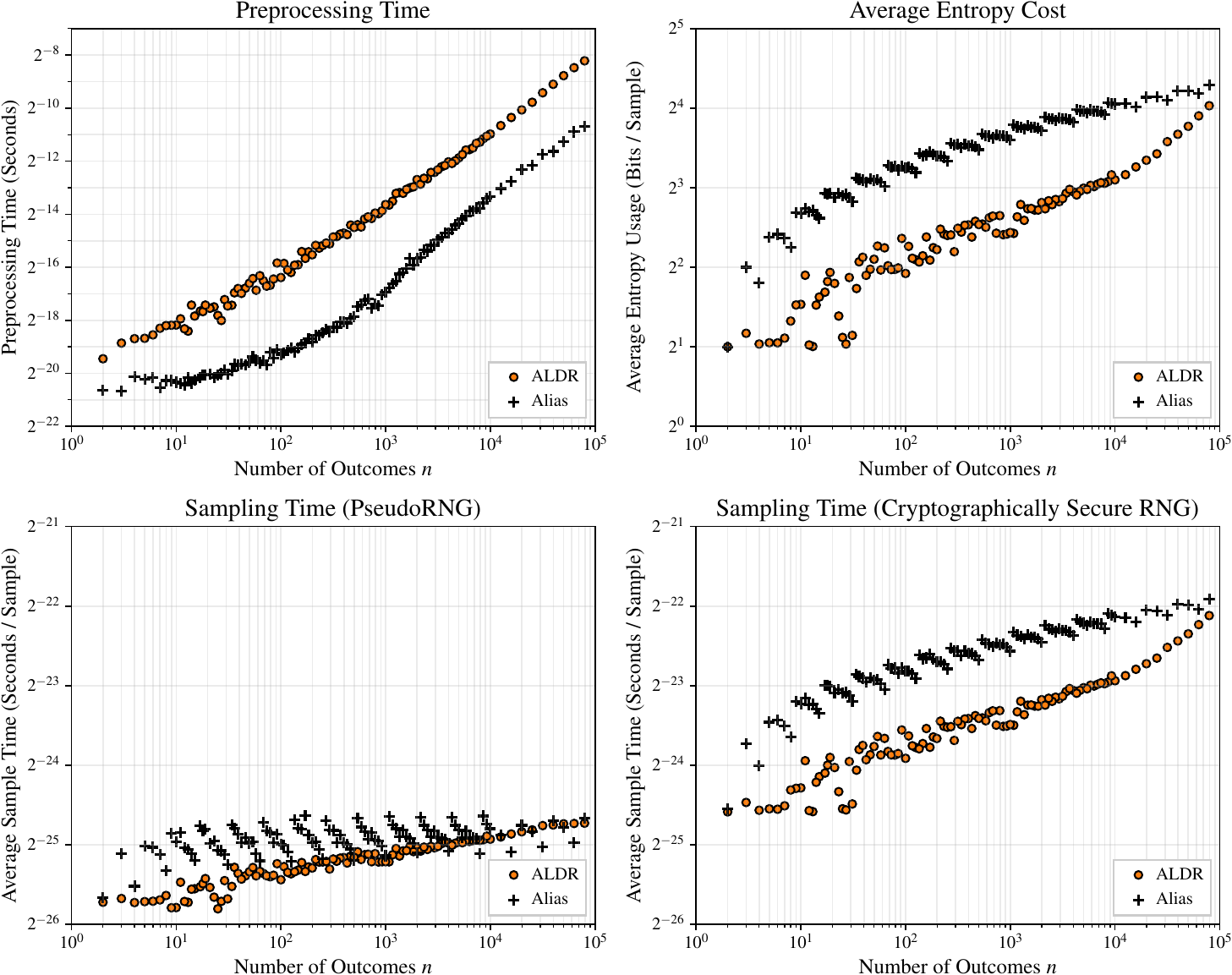}
\caption{Comparison of $\aldr$ with the \citeauthor{walker1977}
alias method~\citep{walker1977} in terms of
preprocessing time (top left panel),
entropy cost (top right panel),
and sampling time using two different pseudorandom number generators (bottom panels).
Each dot shows an $m$-type probability distribution
$\bp = (a_1/m, \dots, a_n/m)$ over $n$ outcomes with $m=10^5$. }
\label{fig:alias-aldr-compare}
\end{figure*}

\subsection{Comparison to the Alias Method}
\label{sec:implementation-alias}

The \citeauthor{walker1977} alias method~\citep{walker1977} is a
state-of-the-art sampling algorithm for discrete probability distributions.
For a target distribution $\bp = (a_1/m, \dots, a_n/m)$, the linear time preprocessing
method of \citet{vose1991} constructs a length-$n$ array $L$,
where each entry $L[i]$
stores a pair of distinct outcomes $y_{i}, z_{i} \in \set{1,\dots,n}$
and a rational weight $w_i \in [0,1]$.
After preprocessing, the sampling phase
for generating $X \sim \bp$ operates as follows:
\begin{itemize}
\item Generate a random index $I \sim \mathrm{Uniform}(1,\dots,n)$.
\item Generate $B \sim \mathrm{Bernoulli}(w_I)$;
  if $B=0$ then set $X \gets y_I$;
  else set $X \gets z_I$.
\end{itemize}
The entropy-optimal sampler for the random index
$I$ requires at most $\ceil{\log(n)}+1$ flips
on average and $O(\log(n))$ space, using the method of \citet{lumbroso2013}.
As the entropy-optimal sampler for $B$ has a cost
upper bounded by 2 flips (cf.~\cref{remark:bernoulli-nu}),
a tight upper bound for the overall entropy cost is
$\ceil{\log(n)} + 3$.

Most software libraries (e.g., \citep{galassi2009,leydold2009})
that implement the alias method
provide approximate implementations with floating-point arithmetic, even
when the input weights $(a_1, \dots, a_n)$ are integers.
A notable exception is the Rust programming language~\citep{rust2014},
which provides an exact (error-free)
implementation of the alias method for integer weights but does not use
entropy-optimal samplers to generate $I$ and $B$.
To ensure a fair comparison, $\aldr$
(\cref{alg:aldr-preprocess,alg:aldr-sample}) and an entropy-optimal
implementation of the alias method were developed%
\footnote{\url{https://github.com/probsys/amplified-loaded-dice-roller}}
in the C programming language
and evaluated on 423 distributions $\bp=(a_1/m,\dots,a_n/m)$ over
$n \in \set{2,\dots,10^5}$ outcomes and $m=10^5$.

\Cref{fig:alias-aldr-compare} shows the preprocessing times, the average
entropy costs, the sampling runtimes using a pseudorandom number generator,
and the sampling runtimes using a cryptographically secure random number generator.
The top-left panel shows that $\aldr$ has a higher preprocessing time
than the alias method, because the former must loop over the $K$
bits of the amplified integer weights.
$\aldr$ enjoys a lower average entropy cost than the alias method,
by virtue of its $H(\bp) + 2$ bound as opposed to $\ceil{\log(n)}+3$.
This gap is most noticeable when $\bp$ has low entropy, i.e.,
$H(\bp) \ll \log(n)$, and shrinks
as $\bp$ becomes closer to a uniform distribution, i.e., $H(\bp) \to \log(n)$.
(For the case $n=2$, the alias method is actually entropy optimal, whereas
$\aldr$ is not necessarily.
But in this case, an entropy-optimal Bernoulli sampler can be used directly,
without the need for any preprocessing.)
These improvements in entropy efficiency translate to direct improvements
in the average sampling time (which will overshadow the difference in
preprocessing time when taking many samples).
When using a fast pseudorandom number generator, the runtime
improvements are most prominent when $n \ll m$
(e.g., $n \leq 10^3 \ll 10^5 = m$ in \cref{fig:alias-aldr-compare}).
When using an expensive cryptographically secure random number generator,
the entropy cost becomes a more significant component of the overall
wall-clock runtime, leading to further improvements of $\aldr$.

\section{Remarks}
\label{sec:remarks}

\paragraph*{Eliminating the Toll Gap}

The \textit{Amplified Loaded Dice Roller} family of samplers
proposed in this work uses
rejection sampling to target an arbitrary $m$-type distribution $\bp$,
by using a proposal distribution
$\bq_K$ whose probabilities are multiples of $1/2^K$.
It has been shown in \cref{theorem:aldr-2k-toll-two}
that this family contains a sampler whose space
complexity scales linearithmically with the size of $\bp$ and whose
entropy cost lies within the entropy-optimal range $[H(\bp), H(\bp)+2)$.
This property in turn ensures that the toll $\toll{\aldr[\bp,2k]} < 2$ has
the same guarantee as the entropy-optimal sampler, for every $\bp$.
\Cref{proposition:aldr-large-gap}, however, shows that, for any $\epsilon > 0$,
there exists a certain distribution $\bp$ for which the toll gap
$\toll{\aldr[\bp,K]} - \toll{\ky[\bp]} > 2 - 2\epsilon$ with respect to the
entropy-optimal sampler fails to converge to zero as $K$ grows (\cref{fig:fldr-toll-1669}).
An open question is whether it is possible to eliminate this
entropy-cost gap while retaining linearithmic space complexity.

\paragraph*{Optimal Amplification Factor}

\Cref{proposition:aldr-not-monotonic-decrease} shows that the entropy
cost of $\aldr[\bp,K]$ is not necessarily monotonic as a function of
$K$ (e.g., \cref{fig:fldr478-tolls}, $\aldr[\bp,10] < \aldr[\bp,11]$).
More generally, our choice of amplification factor $c_K = \floor{2^K/m}$
for the proposal distribution $Q_K \defeq Q_{K,c_K}$~\cref{eq:amplified-proposal-distribution-family}
minimizes the probability of a rejection per rejection trial---and in turn
the expected number of trials---but it does not necessarily minimize the
entropy cost~\cref{eq:fldr-cost} which also accounts
for the entropy cost per trial.
In particular, for a maximum DDG tree depth of $K$,
the optimal amplification factor
\begin{align*}
c^*_{K} \defeq \argmin_{c = 1, 2, \dots, \floor{2^K/m}} \set*{2^{K}/cm \cdot \expect{\cost{\ky[Q_{K,c}]}}}
\end{align*}
can be found in $O(nK^2 2^K/m)$ time by enumeration,
which scales exponentially with the input size when, e.g., $K=2k$.
Is there a more efficient optimization algorithm to find $c^*_K$?
If $c_K$ is odd and $c^*_K$ is even, the rejection sampler with proposal
$Q_{K,c^*_K}$ has smaller depth than $\aldr[\bp,K]$, while matching
or outperforming its entropy cost.

\section*{Acknowledgements}

Comments by the referees are acknowledged.
This material is based upon work supported by the National Science
Foundation under Grant No.~2311983. Any opinions, findings, and conclusions
or recommendations expressed in this material are those of the author(s)
and do not necessarily reflect the views of the National Science
Foundation.

\printbibliography

\begin{IEEEbiographynophoto}{Thomas Draper}
is a Ph.D.~student in Computer Science at Carnegie Mellon University,
Pittsburgh, PA.
He received B.S.~degrees in Mathematics, Computer Science, Statistics, and
Physics from Brigham Young University, Provo, UT in 2024.
His research interests include probabilistic algorithms, information
theory, and quantum computation.
\end{IEEEbiographynophoto}

\begin{IEEEbiographynophoto}{Feras Saad}
is an Assistant Professor of Computer Science at Carnegie Mellon University,
Pittsburgh, PA.
He received the S.B., M.Eng., and Ph.D.\ degrees in Electrical Engineering
and Computer Science from the Massachusetts Institute of Technology (MIT),
Cambridge, MA in 2016, 2016, and 2022, respectively.
From 2022 to 2023, he was a Visiting Scientist at Google.
Dr.~Saad's graduate dissertations were recognized with the Charles \&
Jennifer Johnson Computer Science M.Eng Thesis Award (2017) and George
M.~Sprowls Ph.D.~Thesis Award in Artificial Intelligence and Decision
Making (2023) at MIT.
His research interests include
probabilistic programming languages,
statistical modeling and inference, and
the mathematics of probabilistic computation.
\end{IEEEbiographynophoto}
\cleardoublepage

\onecolumn

\appendices

\setcounter{page}{1}
\renewcommand{\thepage}{S-\arabic{page}}

\section{Proof of \crefnameof{theorem}~\ref{theorem:aldr-2k-toll-two}}
\label[appendix]{appx:proof-theorem-aldr-2k-toll-two}

If $m$ is a power of two then $\aldr[\bp, K]$ is entropy optimal for any $K \geq
k$, and the result follows from \cref{theorem:knuth-yao}, so we assume for the
rest of the proof that $m$ is not a power of two.

\Cref{lemma:aldr-pow-two-toll} shows that the relative FLDR toll for powers
of two is not less than two, so we must bound the total probability
contribution from power-of-two probabilities.
Write $m = 2^u x$ where $x > 1$ is odd.
Then at most $1-2^{-u}$ of the
probability distribution $\bp$ can come from powers of two.
(The fact that not every $p_i$ is a power of two makes use of coprimality
$\gcd(a_1, \dots, a_n) = 1$, unlike \cref{theorem:aldr-generic-bound}, which
applies to arbitrary integer lists but cannot strictly meet the bound of $2$
bits.)
Let $b$ be the smallest integer such that at most $1 - 2^{-b}$ of the
probability comes from powers of two, so $0 \leq b \leq u \leq k-2$ and
the total probability not from powers of two is in
$\set{2^{-b}, 2^{-b}+2^{-u}, \dots, 2^{1-b}-2^{-u}}$.
Then, for every $i$ such that $a_i / m$ is not a power of two,
$a_i \in \set{1,2,3,\dots,m(2^{1-b}-2^{-u})}$.
Dividing by $m$ and multiplying by an appropriate power of two shows that
$1/2 < 2^{-1-\floor{\log(a_i/m)}} a_i/m \leq 1-2^b/m$.
Finally, \cref{lemma:q-p-tightness} shows that
$1/2 < 2^{-1-\floor{\log(a_i/m)}} A_i/2^K \leq (1-2^b/m)M/2^K$,
which matches the conditions to apply the relative toll bounds,
\cref{lemma:relative-toll-bound-linear,corollary:relative-toll-bound-constant}.

If $b=k-2$, then $x=3$ and the relative FLDR toll of power-of-two probabilities
is exactly $2$, and for $a_i < 3$ we have $\trelfldr{A_i, M} = \trel{1/3} < 2$,
so the overall toll is less than $2$ by
\cref{eq:entropy-difference-convex-combination}.
Henceforth, $0 \leq b \leq k-3$.

We now bound the relative FLDR toll of the non-power-of-two probabilities by
proving in three cases that
\begin{equation}
  \textrm{for each $i$, if $p_i$ is not a power of two then } \trel{A_i/2^K} < 2 - (A_0/M) (2^{b+1}+K+1-\floor{\log(A_0)}).
\label{eq:relative-toll-not-pow-two}
\end{equation}

\begin{enumerate}[label={Case \arabic*.}]
\item
If $k \geq 8$ and $2^{-1-\floor{\log(a_i/m)}} A_i/M < 1 - 2^{1+b-k}$, then
\cref{corollary:relative-toll-bound-constant} gives
\begin{align*}
\trel{A_i/2^K}
&< 2 - \frac{k-b-3+1/\ln(2)}{2^{k-b-1}} \\
&= 2 - \frac{1}{2^k}\left(2^{b+1}+2^{b+1}(k-b-4+1/\ln(2))\right).
\end{align*}

Next, we claim that $2^{b+1}(k-b-4+1/\ln(2)) > k + 2$,
which we prove by casework on $0 \leq b \leq k-3$.
If $b = 0$, then
\[
2^{b+1}(k-b-4+1/\ln(2))
= 2(k-4+1/\ln(2))
= k + (k-8) + 2/\ln(2)
> k + 2
\]
because $k \geq 8$.
The same inequality holds for $1 \leq b \leq k-4$ because
\[
2^{b+1}(k-b-4+1/\ln(2)) > 2^{(b-1)+1}(k-(b-1)-4+1/\ln(2))
\]
is an increasing function of $b$ in this range.
Lastly, if $b = k-3$, then
\[
2^{b+1}(k-b-4+1/\ln(2))
= 2^{k-2}(-1 + 1/\ln(2))
> k + 2.
\]

Now we apply $2^{b+1}(k-b-4+1/\ln(2)) > k + 2$ to obtain
\[
\trel{A_i/2^K}
< 2 - \frac{1}{2^k}\left(2^{b+1}+k+2\right).
\]

We claim that
\begin{align*}
\frac{1}{2^k}\left(2^{b+1}+k+2\right)
&\geq \frac{A_0}{M}\left(2^{b+1}+K+1-\floor{\log(A_0)}\right),
\end{align*}
which will complete the proof of this case.
We consider two possibilities for $A_0$.
If $A_0 < 2^{k-1}$, then
\[
\frac{1}{2^k}\left(2^{b+1}+k+2\right)
\geq 2^{K-k+1} \frac{A_0}{M}\left(2^{b+1}+k+2\right)
>\frac{A_0}{M}\left(2^{b+1}+K+1-\floor{\log(A_0)}\right).
\]
If $A_0 \geq 2^{k-1}$, then we have
\[
\frac{1}{2^k}\left(2^{b+1}+k+2\right)
\geq 2^{K-k} \frac{A_0}{M}\left(2^{b+1}+k+2\right)
\geq\frac{A_0}{M}\left(2^{b+1}+K+1-(k-1)\right),
\]
and the case is complete.

\item
If $k \geq 8$ and
$2^{-1-\floor{\log(a_i/m)}} A_i/M \in [1 - 2^{1+b-k}, 1 - 2^{b-k})$,
then \cref{lemma:relative-toll-bound-linear} gives
\begin{align*}
&\quad \trel{A_i / 2^K} \\
&< 2 - \left[ 1 - \left(1 - \frac{2^b}{m}\right) \left(1 - \frac{A_0}{2^K}\right) \right](k-b-2+1/\ln(2)) \\
&= 2 - \left(\frac{2^b}{m}+\frac{A_0}{2^K}-\frac{2^b A_0}{2^K m}\right)(k-b-2+1/\ln(2)) \\
&= 2 - \frac{A_0}{M}\left(\frac{c_K 2^b}{A_0} + \frac{c_K(m-2^b)}{2^K}\right) (k-b-2+1/\ln(2)) \\
&= 2 - \frac{A_0}{M}\left(2^b\left( \frac{c_K}{A_0} - \frac{c_K}{2^K} \right) (k-b-2+1/\ln(2)) + \frac{M}{2^K} (k-b-2+1/\ln(2))\right) \\
&< 2 - \frac{A_0}{M}\left(2^{b+1}+2^b\left( \frac{c_K}{A_0} - \frac{c_K}{2^K} \right) (k-b-4+1/\ln(2)) + \frac{M}{2^K} (k-b-2+1/\ln(2))\right) \\
&< 2 - \frac{A_0}{M}\left(2^{b+1}+2^b \frac{2^{K-k}}{A_0} (k-b-4+1/\ln(2)) + k-b-1\right).
\end{align*}
From here, to prove the desired bound, it suffices to show
\[
2^b \frac{2^{K-k}}{A_0} (k-b-4+1/\ln(2)) + k-b-1
\geq
K+1-\floor{\log(A_0)}.
\]
We proceed by casework on $K$, $A_0$, and $b$.

First, we reduce to the case $K = 2k$ because $2^b (k-b-4+1/\ln(2)) \geq 1$,
which implies
\[
2^b \frac{2^{(K+1)-k}}{A_0} (k-b-4+1/\ln(2))
-
2^b \frac{2^{K-k}}{A_0} (k-b-4+1/\ln(2))
\geq
(K+1) - K
\]
for all $K \geq 2k$, i.e., increasing $K$ tightens the intermediate bound we
have proven faster than it tightens the bound we want to prove.

Next, we split into cases comparing $A_0$ to $2^{k-1}$ and $2^{k-2}$, recalling that $A_0 > 0$.
\begin{enumerate}[label={Case 2.\alph*.}]
\item
If $1 \leq A_0 < 2^{k-2}$, then
\[
2^b \frac{2^k}{A_0} (k-b-4+1/\ln(2)) + k-b-1
\geq 2k+1
\geq 2k+1-\floor{\log(A_0)}.
\]

\item
If $2^{k-2} \leq A_0 < 2^{k-1}$, then
\[
2^b \frac{2^k}{A_0} (k-b-4+1/\ln(2)) + k-b-1
\geq k+3
= 2k+1-\floor{\log(A_0)}.
\]

\item
If $A_0 \geq 2^{k-1}$, then it suffices to show that
$2^b (k-b-4+1/\ln(2)) + k-b-1 \geq k + 2$, or equivalently, that
$2^b (k-b-4+1/\ln(2)) \geq b + 3$.
To this end, we again use casework, now on $b$, recalling $k \geq 8$.
\begin{enumerate}[label={Case 2.c.\roman*.}]
\item
If $b=0$, then
\[
2^b (k-b-4+1/\ln(2)) = k-4+1/\ln(2) > 3.
\]

\item
If $b=1$, then
\[
2^b (k-b-4+1/\ln(2)) = 2(k-5+1/\ln(2)) > 4.
\]

\item
If $b=2$, then
\[
2^b (k-b-4+1/\ln(2)) = 4(k-6+1/\ln(2)) > 5.
\]

\item
If $b=3$, then
\[
2^b (k-b-4+1/\ln(2)) = 8(k-7+1/\ln(2)) > 6.
\]

\item
If $b \geq 4$, then
\[
2^b (k-b-4+1/\ln(2)) \geq 2^b(-1+1/\ln(2)) > b + 3.
\]
\end{enumerate}
\end{enumerate}

This analysis completes the proof of
\begin{equation*}
\trel{A_i / 2^K} < 2 - \frac{A_0}{M}\left(2^{b+1}+K+1-\floor{\log(A_0)}\right)
\end{equation*}
for the case $k \geq 8$ and
$2^{-1-\floor{\log(a_i/m)}} A_i/M \in [1 - 2^{1+b-k}, 1 - 2^{b-k})$.

\item
If $k < 8$, then direct enumeration
shows that
\cref{eq:relative-toll-not-pow-two} holds for all
$0 \leq b \leq k-3$ and $2k \leq K < 16$:
\begin{equation*}
\min_{\substack{
  2k \leq K < 16; \\
  a_i / m \textrm{ not a power of two}
}}
\left[ 2 - (A_0/M) (2^{b+1}+K+1-\floor{\log(A_0)}) - \trel{A_i/2^K} \right]
\gtrapprox 0.0394,
\end{equation*}
and this tightest case occurs at $a_i = 117$, $m = 118$, and $K = 14$.
For $K \geq 16$, we reduce to the case $k = 8$, by repeating the array of
weights $2^{8-k}$ times and amplifying these repeated weights using the depth
$K+8-k \geq 16$, which gives the amplified weight sum $2^{8-k}M$ and amplified
rejection weight $2^{8-k}A_0$.
From the result of the previous cases, together with
\cref{lemma:bit-shift-relative-toll},
it follows that
\begin{align*}
\trel{A_i/2^K}
= \trel{A_i/2^{K+8-k}}
&< 2 - \frac{2^{8-k}A_0}{2^{8-k}M}\left(2^{b+1}+K+8-k+1-\floor{\log(2^{8-k}A_0)}\right) \\
&= 2 - \frac{A_0}{M}\left(2^{b+1}+K+1-\floor{\log(A_0)}\right),
\end{align*}
which completes the proof of \cref{eq:relative-toll-not-pow-two}.
\end{enumerate}

Applying \cref{lemma:aldr-rejection-toll} together with
\cref{eq:relative-toll-not-pow-two},
if $a_i/m$ is not a power of two, then
\begin{align}
\trelfldr{A_i, M}
&= \trel{A_i / 2^K}
  + \frac{2^K}{M} \left( \Ho{M/2^K} + \nuu{1 - M/2^K} \right)
  \notag \\
&< 2 - \frac{A_0}{M}\left(2^{b+1}+K+1-\floor{\log(A_0)}\right)
  + \frac{A_0}{M} (K+3-\floor{\log(A_0)})
  \notag \\
&< 2 - (2^{b+1} - 2) A_0 / 2^K.
\label{eq:relative-aldr-toll-not-pow-two}
\end{align}

Now, we use \cref{lemma:aldr-pow-two-toll} and
\cref{eq:relative-aldr-toll-not-pow-two} in
\cref{eq:entropy-difference-convex-combination} to bound the overall toll.
In the worst case, $1-2^{-b}$ of the probability comes
from powers of two, which gives the bound
\begin{align*}
\toll{\aldr[\bp, K]}
&= \mathbb{E}_{i \sim \bp}\left[ \trelfldr{A_i, M} \right] \\
&< (1-2^{-b}) 2 (1 + A_0 / 2^K)
  + 2^{-b} \left( 2 - (2^{b+1} - 2) A_0 / 2^K \right) \\
&= 2.
\end{align*}
The claim that $\toll{\aldr[\bp, K]} < 2$ for all $\bp$ and all
$K \geq 2k$ is thus established.

\section{Proof of \crefnameof{theorem}~\ref{theorem:generating-fn}}
\label[appendix]{appx:proof-theorem-generating-fn}

Because $\ky[\bq]$ is entropy optimal,
$\leaves{\ky[\bq]}{d}{i} = \epsd{A_i/2^K}$ for all $d$ and
$0 \leq i \leq n$.
Upon unrolling one back edge at depth $r$, the label counts at depth $d$ become
$\epsd{A_i/2^K} + \epsd[d-r]{A_i/2^K}$, because one new leaf node with
label $i$ appears at depth $d+r$ for every leaf with label $i$ at depth $d$ in
the original tree.
Therefore, if there is just a single back edge at depth $r$, then by repeatedly
unrolling it to get an infinite tree, we have
\begin{equation*}
\leaves{\fldr*[A_1,\dots,A_n]}{d}{i}
= \epsd{A_i/2^K} + \epsd[d-r]{A_i/2^K} + \epsd[d-2r]{A_i/2^K} + \cdots,
\end{equation*}
which confirms the theorem for the case $A_0/2^K = 2^{-r}$.
This case is revisited in the proof of \cref{theorem:aldr-ky-match-bounded}.

To analyze the more general case where there may be multiple back edges in the
FLDR tree, it will be convenient to introduce some notation to describe the
convolutions that arise from unrolling multiple back edges at once.
For this purpose, we will use generating functions.
We wish to construct the formal power series
$\sum_{d=0}^{\infty} \leaves{\fldr*[A_1,\dots,A_n]}{d}{i} z^d$,
whose $z^d$-coefficient is the number of leaves with label $i$ at depth $d$
in the unrolled tree $\fldr[A_1,\dots,A_n]$.
The generating function for the leaves with label $i$ in the unexpanded FLDR tree is
\begin{equation*}
\sum_{d=0}^{\infty} \leaves{\ky[\bq]}{d}{i} z^d
= \sum_{d=0}^{K} \epsd{A_i/2^K} z^d.
\end{equation*}

Upon unrolling back edges targeting the root of the tree, the new leaves at
depth $d$ are in bijective correspondence with pairs of back edge sources at
depth $r$ and leaves at depth $d-r$, which corresponds exactly to the
operation of convolution or polynomial multiplication.
The generating function describing the newly-added leaves is therefore
\begin{equation*}
\sum_{d=0}^\infty \sum_{r=0}^\infty \epsd[r]{A_0/2^K} \epsd[d-r]{A_i/2^K} z^d
= \sum_{d=0}^{K} \epsd{A_0/2^K} z^d \sum_{d=0}^{K} \epsd{A_i/2^K} z^d.
\end{equation*}
Therefore, upon unrolling all back edges present in the original tree once
(cf.~\cref{fig:fldr-tree-unrolling-1}), the generating function for the leaves
with label $i$ becomes
\begin{equation*}
\sum_{d=0}^{K} \epsd{A_i/2^K} z^d
+ \sum_{d=0}^{K} \epsd{A_i/2^K} z^d \sum_{d=0}^{K} \epsd{A_0/2^K} z^d
= \left(1 + \sum_{d=0}^{K} \epsd{A_0/2^K} z^d\right) \sum_{d=0}^{K} \epsd{A_i/2^K} z^d,
\end{equation*}
and the generating function for the rejection labels in the new tree is
$(\sum_{d=0}^{K} \epsd{A_0/2^K} z^d)^2$.
Upon unrolling these new back edges once (cf.~\cref{fig:fldr-tree-unrolling-2}),
the generating function for the leaves with label $i$ becomes
\begin{equation*}
\left(
  1
  + \sum_{d=0}^{K} \epsd{A_0/2^K} z^d
  + \left(\sum_{d=0}^{K} \epsd{A_0/2^K} z^d\right)^2
  \right)
  \sum_{d=0}^{K} \epsd{A_i/2^K} z^d,
\end{equation*}
and the new rejection labels have generating function
\begin{equation*}
\left(\sum_{d=0}^{K} \epsd{A_0/2^K} z^d\right)^3.
\end{equation*}
By inductively continuing this pattern, the generating function for
the leaves with label $i$ in the fully-unrolled infinite tree is
\begin{align*}
\sum_{d=0}^{\infty} \leaves{\fldr*[A_1,\dots,A_n]}{d}{i} z^d
&=
\left(
  1
  + \sum_{d=0}^{K} \epsd{A_0/2^K} z^d
  + \left(\sum_{d=0}^{K} \epsd{A_0/2^K} z^d\right)^2
  + \cdots
  \right)
\sum_{d=0}^{K} \epsd{A_i/2^K} z^d,
\end{align*}
which completes the proof.

\end{document}